\documentclass{article}
\usepackage[utf8]{inputenc}

\title{Isoperimetric Inequalities for Real-Valued Functions with Applications to Monotonicity Testing}
\author{
Hadley Black\thanks{Department of Computer Science, University of California, Los Angeles. Email: \href{mailto:hablack@cs.ucla.edu}{\nolinkurl{hablack@cs.ucla.edu}}. This work was supported by NSF Grant
CCF-1553605 and Boston University's Data Science Initiative.}
\and 
Iden Kalemaj\thanks{Department of Computer Science, Boston University. Email: \href{mailto:ikalemaj@bu.edu}{\nolinkurl{ikalemaj@bu.edu}}. This work was supported by NSF award CCF-1909612 and Boston University's Dean's Fellowship.}
\and 
Sofya Raskhodnikova\thanks{Department of Computer Science, Boston University. Email: \href{mailto:sofya@bu.edu}{\nolinkurl{sofya@bu.edu}}. This work was supported by NSF award CCF-1909612.}
\date{\vspace{-5ex}}
\date{}  
}

\usepackage{amsmath,amssymb,amsthm,verbatim,relsize,mathrsfs,hyperref,soul}
\usepackage{algorithm}
\usepackage[noend]{algpseudocode}
\hypersetup{colorlinks = true,linkcolor = blue, anchorcolor = blue, citecolor = blue, filecolor = red,urlcolor = red}
\usepackage [english]{babel}
\usepackage [autostyle, english = american]{csquotes}
\usepackage{blindtext}
\usepackage{framed}
\usepackage{fullpage}
\usepackage[shortlabels]{enumitem}
\usepackage{graphicx}
\usepackage{xcolor}
\usepackage{subcaption}
\usepackage{subfiles}
\usepackage{dsfont}
\usepackage{bbm}

\usepackage{setspace}

\MakeOuterQuote{"}

\makeatletter
\newcommand\footnoteref[1]{\protected@xdef\@thefnmark{\ref{#1}}\@footnotemark}
\makeatother

\newcommand{\ignore}[1]{}

\newcommand{\cD}{\mathcal{D}}

\newcommand{\cQ}{\mathcal{Q}}
\newcommand{\cS}{\mathcal{S}}

\newcommand{\hM}{\widehat{M}}

\newcommand{\R}{\mathbb R}

\newcommand{\eps}{\varepsilon}

\newcommand{\calH}{{\cal H}}

\newcommand{\cG}{{\mathcal{G}}}
\newcommand{\cH}{{\mathcal{H}}}

\newcommand{\bx}{\boldsymbol{x}}
\newcommand{\by}{\boldsymbol{y}}

\newcommand{\bS}{\boldsymbol{S}}

\newcommand{\bT}{\boldsymbol{T}}

\newcommand{\NN}{\mathbb{N}}
\newcommand{\RR}{\mathbb{R}}

\newcommand{\sfm}{\mathcal{S}_f^-}

\newcommand{\red}{\mathrm{red}}
\newcommand{\blue}{\mathrm{blue}}

\newcommand{\floor}[1]{\lfloor#1\rfloor}

\newcommand{\EX}{\mathbb{E}}

\newcommand*\Let[2]{\State #1 $\gets$ #2}
\algrenewcommand\algorithmicrequire{\textbf{Input:}}
\algrenewcommand\algorithmicensure{\textbf{Output:}}

\newcommand{\dist}{\mathtt{dist}}

\newtheorem{theorem}{Theorem}[section]

\newtheorem{definition}[theorem]{Definition}
\newtheorem{lemma}[theorem]{Lemma}
\newtheorem{corollary}[theorem]{Corollary}

\newtheorem{observation}[theorem]{Observation}
\newtheorem{fact}[theorem]{Fact}
\newtheorem{claim}[theorem]{Claim}

\newcommand{\Sec}[1]{\hyperref[sec:#1]{Section~\ref*{sec:#1}}} 
\newcommand{\Eqn}[1]{\hyperref[eq:#1]{(\ref*{eq:#1})}} 
\newcommand{\Fig}[1]{\hyperref[fig:#1]{Fig.\,\ref*{fig:#1}}} 
\newcommand{\Tab}[1]{\hyperref[tab:#1]{Tab.\,\ref*{tab:#1}}} 
\newcommand{\Thm}[1]{\hyperref[thm:#1]{Theorem\,\ref*{thm:#1}}} 
\newcommand{\Fact}[1]{\hyperref[fact:#1]{Fact\,\ref*{fact:#1}}} 
\newcommand{\Lem}[1]{\hyperref[lem:#1]{Lemma\,\ref*{lem:#1}}} 
\newcommand{\Prop}[1]{\hyperref[prop:#1]{Prop.~\ref*{prop:#1}}} 
\newcommand{\Cor}[1]{\hyperref[cor:#1]{Corollary~\ref*{cor:#1}}} 
\newcommand{\Conj}[1]{\hyperref[conj:#1]{Conjecture~\ref*{conj:#1}}} 
\newcommand{\Def}[1]{\hyperref[def:#1]{Definition~\ref*{def:#1}}} 
\newcommand{\Alg}[1]{\hyperref[alg:#1]{Algorithm~\ref*{alg:#1}}} 
\newcommand{\Step}[1]{\hyperref[step:#1]{Step.~\ref*{step:#1}}} 
\newcommand{\Ex}[1]{\hyperref[ex:#1]{Ex.~\ref*{ex:#1}}} 
\newcommand{\Clm}[1]{\hyperref[clm:#1]{Claim~\ref*{clm:#1}}} 
\newcommand{\Inv}[1]{\hyperref[inv:#1]{Invariant~\ref*{inv:#1}}} 
\newcommand{\Rem}[1]{\hyperref[rem:#1]{Remark~\ref*{rem:#1}}} 
\newcommand{\Obs}[1]{\hyperref[obs:#1]{Observation~\ref*{obs:#1}}} 

\begin{document}
\maketitle

\begin{abstract}
We generalize the celebrated isoperimetric inequality of Khot, Minzer, and Safra~(SICOMP 2018) for Boolean functions to the case of real-valued functions $f:\{0,1\}^d\to\RR$.
Our main tool in the proof of the generalized inequality is a new Boolean decomposition that represents every real-valued function $f$ over an arbitrary partially ordered domain as a collection of Boolean functions over the same domain, roughly capturing the distance of $f$ to monotonicity and the structure of violations of $f$ to monotonicity.

We apply our generalized isoperimetric inequality to improve algorithms for testing monotonicity and approximating the distance to monotonicity for real-valued functions. Our tester for monotonicity has query complexity $\widetilde{O}(\min(r \sqrt{d},d))$, where $r$ is the size of the image of the input function. (The best previously known tester, by Chakrabarty and Seshadhri (STOC 2013), makes $O(d)$ queries.) Our tester is nonadaptive and has 1-sided error. We show a matching lower bound for nonadaptive, 1-sided error testers for monotonicity. We also show that the distance to monotonicity of real-valued functions that are $\alpha$-far from monotone can be approximated nonadaptively within a factor of $O(\sqrt{d\log d})$ with query complexity polynomial in  $1/\alpha$ and the dimension $d$. This query complexity is known to be nearly optimal for nonadaptive algorithms even for the special case of Boolean functions. (The best previously known distance approximation algorithm for real-valued functions, by Fattal and Ron (TALG 2010) achieves $O(d\log r)$-approximation.)
\end{abstract}


\section{Introduction}\label{sec:intro}

We investigate the structure of real-valued functions over the domain $\{0,1\}^d$, the $d$-dimensional hypercube. Our main contribution is a generalization of a powerful tool from the analysis of Boolean functions, specifically, isoperimetric inequalities\footnote{We discuss isoperimetric inequalities that study the size of the ``boundary'' between the points on which the function takes value 0 and the points on which it takes value 1. The boundary size is defined in terms of the edges of the $d$-dimensional hypercube with vertices labeled by the values of the function. The edges of the hypercube might be directed or undirected, depending on the type of the inequality.}, to the case of real-valued functions. Isoperimetric inequalities for the undirected hypercube were studied by Margulis~\cite{Mar74} and Talagrand~\cite{Tal93}. Chakrabarty and Seshadhri~\cite{CS16} had a remarkable insight to develop a directed analogue of the Margulis inequality. This beautiful line of work culminated in the directed analogue of the Talagrand inequality proved by Khot, Minzer, and Safra~\cite{KhotMS18}. We refer to this as the KMS inequality. As Khot, Minzer, and Safra\ explain in their celebrated work, the Margulis-type inequalities follow from the Talagrand-type inequalities and, more generally, the directed analogue of the Talagrand inequality implies all the other inequalities we mentioned.
We generalize all these inequalities to the case of real-valued functions.

For the directed case, we prove a generalization of the KMS inequality for functions  $f \colon \{0,1\}^d \to \mathbb{R}$.
To generalize the \emph{undirected} isoperimetric inequalities, we give a property testing interpretation of the Talagrand inequality. With this interpretation, the inequality easily generalizes to the case of real-valued functions.

Our proofs of the new isoperimetric inequalities reduce the general case to the Boolean case. Our main tool for generalizing the KMS inequality is a new Boolean decomposition theorem that represents every real-valued function $f$ over an arbitrary partially ordered domain as a collection of Boolean functions over the same domain, roughly capturing the distance of $f$ to monotonicity and the structure of violations of $f$ to monotonicity. 

We apply our generalized isoperimetric inequality to improve algorithms for testing monotonicity and approximating the distance to monotonicity for real-valued functions.
Our algorithm for testing monotonicity is nonadaptive and has 1-sided error. An algorithm is {\em nonadaptive} if its input queries do not depend on answers to previous queries. A property testing algorithm has {\em 1-sided error} if it always accepts all inputs with the property it is testing. We show that our algorithm for testing monotonicity is optimal among nonadaptive, 1-sided error testers. Our distance approximation algorithm is nonadaptive. Its query complexity is known to be nearly optimal for nonadaptive algorithms, even for the special case of Boolean functions.

\subsection{Isoperimetric Inequalities for Real-Valued Functions}\label{sec:intro-inequalities}

We view the domain of functions $f \colon \{0,1\}^d\to \RR$ as a hypercube. For the directed isoperimetric inequalities, the edges of the hypercube are ordered pairs $(x,y)$, where $x,y\in \{0,1\}^d$ and there is a unique\footnote{Given a positive integer $\ell \in \mathbb{Z}^+$, we let $[\ell]$ denote the set $\{1, 2, \dots, \ell\}$.} $i\in[d]$ such that $x_i=0,y_i=1$, and $x_j=y_j$ for all coordinates $j\in[d]/\{i\}$. This defines a natural partial order on the domain: $x \preceq y$ if $x_i \leq y_i$ for all coordinates $i \in [d]$ or, equivalently, if there is a directed path from $x$ to $y$ in the hypercube.
A function $f\colon \{0,1\}^d \to \R$ is \emph{monotone} if $f(x) \leq f(y)$ whenever $x \preceq y$.
The distance to monotonicity of a function $f\colon \{0,1\}^d \to \R$, denoted $\eps(f)$, is the minimum of $\Pr_{x\in\{0,1\}^d}[f(x)\neq g(x)]$  over all monotone functions $g\colon \{0,1\}^d \to \R$.
 An edge $(x, y)$ is \emph{violated} by $f$ if $f(x) > f(y)$. Let $\sfm$ be the set of violated edges. For $x \in \{0,1\}^d$, let $I_f^-(x)$ be the number of $\textit{outgoing}$ violated edges incident on $x$, specifically,
$$
I_f^-(x) = \left|\left\{ y \colon (x,y) \in \cS_f^-\right\}\right|.
$$
Our main result is the following isoperimetric inequality.
\begin{theorem}[Isoperimetric Inequality]
\label{thm:directed_talagrand_real}
There exists a constant $C > 0$, such that for all functions $f\colon \{0,1\}^d \rightarrow \RR$,
\begin{equation}\label{eq:talagrand}
\underset{\bx \sim \{0,1\}^d}{\EX} \biggl[\sqrt{I_f^-(\bx)}\biggr] \geq C \cdot \eps(f).
\end{equation}
\end{theorem}

\Thm{directed_talagrand_real} is a generalization of the celebrated inequality of Khot, Minzer, and Safra~\cite{KhotMS18}, that was subsequently strengthened by Pallavoor et al.~\cite{PallavoorRW20}, who proved (\ref{eq:talagrand}) for the special case of Boolean functions $f \colon \{0,1\}^d \to \{0,1\}$. We show that the same inequality holds for real-valued functions without any dependence on the size of the image of the function. In addition, the constant $C$ is only a factor of 2 smaller than the constant in the inequality of Pallavoor et al.

Applications to monotonicity testing and distance approximation rely on a stronger, ``robust'' version of \Thm{directed_talagrand_real}. The robust version considers an arbitrary 2-coloring $\mathtt{col} \colon \cS_f^- \rightarrow \{\mathrm{red}, \mathrm{blue}\}$ of the violated edges. The color of an edge is used to specify whether the edge is counted towards the lower or the upper endpoint.  Let $I_{f, \mathrm{red}}^-(x)$ be the number of \textit{outgoing} red violated edges incident on $x$, and $I_{f, \mathrm{blue}}^-(x)$ be the number of \textit{incoming} blue violated edges incident on $x$, specifically,
\begin{align*}
    I_{f,\text{red}}^-(x) =& \left|\left\{ y \colon (x,y) \in \cS_f^- \text{, } \mathtt{col}(x,y) = \text{red}\right\}\right|; \\
    I_{f,\text{blue}}^-(y) =& \left|\left\{ x \colon (x,y) \in \cS_f^- \text{, } \mathtt{col}(x,y) = \text{blue}\right\}\right| \text{.}
\end{align*}
Our next theorem is a generalization of the robust isoperimetric inequality for Boolean functions established by Khot, Minzer, and Safra\ and strengthened by Pallavoor et al. As before, the constant $C$ is only a factor of 2 smaller for the real-valued case than for the Boolean case.

\begin{theorem}[Robust Isoperimetric Inequality]\label{thm:real_tal_robust}
There exists a constant $C > 0$, such that for all functions $f \colon \{0,1\}^d \to \RR$ and colorings $\mathtt{col} \colon \cS_f^- \to \{\red,\blue\}$, 
\[
\underset{\bx \sim \{0,1\}^d}{\EX}\left[\sqrt{I_{f,\red}^-(\bx)}\right] + \underset{\by \sim \{0,1\}^d}{\EX}\left[\sqrt{I_{f,\blue}^-(\by)}\right] \geq C \cdot \eps(f).
\]
\end{theorem}

Note that \Thm{real_tal_robust} implies \Thm{directed_talagrand_real} by considering the coloring where all violated edges are red. Therefore, we only present a proof of \Thm{real_tal_robust}.

\paragraph{Boolean decomposition.}
Our main technical contribution is the Boolean decomposition (\Thm{main1}). It allows us to prove \Thm{real_tal_robust} by reducing the general case of real-valued functions to the special case of Boolean functions.
\Thm{main1} states that every non-monotone function $f$ can be decomposed into Boolean functions $f_1, f_2, \dots, f_k$ that collectively preserve the distance to monotonicity of $f$ and violate a subset of the edges violated by $f$.  Crucially, they violate edges in disjoint subgraphs of the hypercube.

Our Boolean decomposition works for functions over any partially ordered domain. We represent such a domain by a directed acyclic graph (DAG). For a DAG $\cG$, we denote its vertex set by $V(\cG)$ and its edge set by $E(\cG).$ A DAG $\cG$ determines a natural partial order on its vertex set: for all $x,y\in V(\cG),$ we have $x\preceq y$  if and only if $\cG$ contains a path from $x$ to $y$. A function $f \colon V(\cG) \rightarrow \RR$  is {\em monotone} if $f(x)\leq f(y)$ whenever $x \preceq y$. An edge $(x,y)$ of $\cG$ is {\em violated} by $f$ if $f(x)>f(y)$.
The definitions of $\eps(f)$, the distance of $f$ to monotone, and $\sfm$, the set of violated edges, are the same as for the special case of the hypercube.

\begin{theorem} [Boolean Decomposition] \label{thm:main1} Suppose $\cG$ is a DAG and $f \colon V(\cG) \to \RR$ is a non-monotone function over the vertices of $\cG$. Then, for some $k \geq 1,$ there exist Boolean functions $f_1,\ldots,f_k \colon V(\cG) \to \{0,1\}$ and disjoint (induced) subgraphs $\cH_1,\ldots,\cH_k$ of $\cG$ for which the following hold:
\begin{enumerate}[noitemsep]
    \item $2\sum_{i=1}^k \eps(f_i) \geq \eps(f)$.
    \item $\cS_{f_i}^- \subseteq \cS_{f}^- \cap E(\cH_i)$ for all $i \in [k]$.
\end{enumerate}
\end{theorem}

%
%
We derive \Thm{real_tal_robust} from \Thm{main1} in \Sec{directed-real-talagrand} and prove \Thm{main1} in \Sec{decomposition}.

A natural first attempt to proving \Thm{directed_talagrand_real} is to try reducing to the special case of Boolean functions (the KMS inequality) via a thresholding argument. Given $f \colon \{0,1\}^d \to \RR$ and $t \in \RR$, define $h_t \colon \{0,1\}^d \to \{0,1\}$ to be $h_t(x) = 1$ iff $f(x) > t$. Clearly, this can only reduce the left-hand side of \Eqn{talagrand} 
since the influential edges of $h_t$ are a subset of the influential edges of $f$. Thus, if there exists some $t \in \R$ such that $\eps(h_t) = \Omega(\eps(f))$, then applying the KMS inequality to $h_t$ would show that the inequality also holds for $f$. In fact, as we show in \Sec{undirected_real}, this technique easily allows us to reduce the \emph{undirected} inequality for the real-valued case to the Boolean case, without any significant additional ideas. However, in the directed setting, a simple argument shows that there exists $f$ for which $\eps(h_t) \leq \eps(f)/r$ for all $t \in \R$, where $r$ is the size of the image of $f$. Thus, additional ideas are required to prove \Thm{directed_talagrand_real} by a reduction to the KMS inequality. The highly structured decomposition of \Thm{main1} gives a collection of disjoint subgraphs $\cH_1,\ldots,\cH_k$ of the directed hypercube where, in each $\cH_i$, an independent "variable thresholding rule" can be applied, yielding the Boolean function $f_i$. The "threshold" for each vertex $x$ in $\cH_i$ depends on the values of the function at a particular set of vertices reachable from $x.$


The Boolean decomposition is quite powerful: in addition to enabling us to prove the new isoperimetric inequality, 
it 
can be used to easily derive a lower bound on the number of edges violated by a real-valued function directly from the bound for the Boolean case, without relying on \Thm{real_tal_robust}. This bound is used to analyze the edge tester for monotonicity whose significance is described in \Sec{applications}. The early works on monotonicity testing~\cite{GGLRS00,DGLRRS99,Ras99} have shown that $|\sfm|\geq \eps(f)\cdot 2^d$ for every Boolean function $f$ on the domain $\{0,1\}^d$. In other words, the number of edges violated by $f$ is at least the number of points on which the value of the function has to change to make it monotone.
This bound was generalized to the case of real-valued functions by~\cite{DGLRRS99,Ras99} who showed that
$|\sfm|\geq (\eps(f)/\lceil\log r\rceil) \cdot 2^d$ for every real-valued function $f$ on the domain $\{0,1\}^d$ and with image size $r$. (The size of the image of $f$ is the number of distinct values it takes.)
Chakrabarty and Seshadhri~\cite{CS13} improved this bound by a factor of $\Theta(\log r)$, thus removing the dependence on the size of the image of the function.
 Our Boolean decomposition of a real-valued function $f$ in terms of Boolean functions $f_1,\dots,f_k$, given by \Thm{main1}, yields this result of~\cite{CS13} as an immediate corollary of the special case for Boolean functions:
$$|\sfm|\geq \sum_{i=1}^k|{\mathcal{S}_{f_i}^-} |
\geq \sum_{i=1}^k \eps(f_i)\cdot 2^d
\geq \eps(f)\cdot 2^{d-1},
$$
where the inequalities follow by first applying Item~2 of \Thm{main1}, then applying the bound for the Boolean case, and, finally, applying Item~1 of \Thm{main1}.

\paragraph{Undirected isoperimetric inequality for real-valued functions.} The original isoperimetric inequality of Talagrand~$\cite{Tal93}$ treats the domain $\{0,1\}^d$ as an undirected hypercube. An undirected edge $\{x,y\}$ is {\em influential} if $f(x)\neq f(y)$. Let $I_f(x)$ be the number of influential edges 
$\{x,y\}$ incident on $x \in \{0,1\}^d$ for which $f(x) > f(y)$. This definition ensures that each influential edge is counted towards $I_f(x)$ for exactly one vertex $x$. The variance $\mathrm{var}(f)$ of a Boolean function is defined as $p_0(1-p_0)$, where $p_0$ is the probability that $f(x)=0$ for a uniformly random point $x$ in the domain. Talagrand \cite{Tal93} proved the following.

\begin{theorem}[Talagrand Inequality \cite{Tal93}]\label{thm:undirected_boolean}
For all functions $f\colon \{0,1\}^d \to \{0,1\}$,
\begin{equation}\label{eq:undirected-talagrand}
\underset{\bx \sim \{0,1\}^d}{\EX} \biggl[\sqrt{I_f(\bx)}\biggr] \geq \sqrt{2}\ \mathrm{var}(f).
\end{equation}
\end{theorem}

Before generalizing \Thm{undirected_boolean} to real-valued functions, we reinterpret it using a property testing notion. Observe that the natural definition of the variance of a real-valued function results in a quantity that depends on specific values of the function, whereas whether an edge is influential depends only on whether the values on its endpoints are different and not on the specific values themselves. So, variance is not a suitable notion for generalizing this inequality. We replace the variance of $f$ with the distance of $f$ to constant, denoted $\dist(f,\mathbf{const})$, i.e., the minimum of $\Pr_{\bx \sim \{0,1\}^d}[f(\bx)\neq g(\bx)]$ over all constant functions $g\colon \{0,1\}^d \to \R$.
For a Boolean function $f$, the distance to  constant is $\min\{p_0,(1-p_0)\}$ and, therefore,
the left-hand side of \Eqn{undirected-talagrand} is at least $\dist(f, \mathbf{const})/\sqrt{2}$.
Next, we state our generalization of Talagrand's inequality, proved in   \Sec{undirected_real}.
\begin{theorem}[Undirected Isoperimetric Inequality]\label{thm:undirected_real}
For all functions $f \colon \{0,1\}^d \to \RR$,
\[
\underset{\bx \sim \{0,1\}^d}{\EX}\left[\sqrt{I_f(\bx)}\right] \geq 
\frac{\dist(f,\mathbf{const})}{2\sqrt{2}}\text{.}
\]
\end{theorem}
We do not discuss Margulis-type isoperimetric inequalities here, but note that their natural generalizations to the real range follow from our Talagrand-type inequalities (for the same reasons as for the special case of Boolean functions, as discussed in~\cite{KhotMS18}).

\subsection{Applications of the New Isoperimetric Inequality for Real-Valued Functions}\label{sec:applications}
We apply our generalized isoperimetric inequality (\Thm{real_tal_robust}) to improve algorithms for testing monotonicity and approximating the distance to monotonicity for real-valued functions.

\paragraph{Monotonicity testing.} Monotonicity  of functions, first studied in the context of property testing by Goldreich et al.~\cite{GGLRS00}, is one of the most widely investigated properties in this model~\cite{EKKRV00,DGLRRS99,Ras99,LR01,FLNRRS02,AC06,Fis04,HK08,BRW05,PRR06,ACCL07,BGJRW12,BCGM12,BBM12,CS13,ChSe14,CS16,BlaRY14,CDJS17, ChenDST15,BB16,ChenWX17,PRV18,BlackCS18, KhotMS18,CS19,BlackCS20}.
A function is $\eps$-far from monotone if its distance to monotonicity is at least $\eps$; otherwise, it is $\eps$-close to monotone.
An $\eps$-tester for  monotonicity is a randomized algorithm that, given a parameter $\eps\in(0,1)$ and oracle access to a function $f$, accepts with probability at least 2/3 if $f$ is monotone and rejects with probability at least 2/3 if $f$ is $\eps$-far from monotone.
Prior to our work, the best monotonicity tester for real-valued functions was the {\em edge tester}. The edge tester, introduced by \cite{GGLRS00}, queries the values of $f$ on the endpoints of uniformly random edges of the hypercube and rejects if it finds a violated edge. As we discussed in \Sec{intro-inequalities}, a series of works \cite{GGLRS00,DGLRRS99,Ras99,CS13} proved lower bounds on $|\sfm|$, the number of violated edges, resulting in the tight analysis of the edge tester for both Boolean and real-valued functions: $O(d/\eps)$ queries are sufficient (and also necessary, e.g., for $f(x)=1-x_1$, the anti-dictator function).
%
For many years, it remained open whether an $o(d)$-query tester for monotonicity existed, until a sequence of breakthroughs~\cite{CS16,ChenST14,KhotMS18} designed testers for Boolean functions with query complexity $\widetilde{O}(d^{7/8}), \widetilde{O}(d^{5/6})$, and finally $\widetilde{O}(\sqrt{d})$. Prior to our work, the same question remained open for functions with image size, $r$, greater than 2.

We show that when $r$ is small compared to $d$, monotonicity can be tested with $o(d)$ queries. (Note that $r\leq 2^d$.)


\begin{theorem}\label{thm:tester} 
There exists a nonadaptive, 1-sided error $\eps$-tester for monotonicity of  $f \colon \{0,1\}^d \to \RR$ 
that makes $\widetilde{O}\Big(\min\big(\frac{r \sqrt{d}}{\eps^2},\frac{d}{\eps}\big)\Big)$ queries and works for all functions $f$ with image size $r$.
\end{theorem}

The proof of \Thm{tester} (in \Sec{tester}) heavily relies on the generalized isoperimetric inequality of \Thm{real_tal_robust}. We extend several other combinatorial properties of Boolean functions to real-valued functions. In particular, the persistence of a vertex $x \in \{0,1\}^d$ is a key combinatorial concept in the analysis. A vertex $x \in \{0,1\}^d$ is $\tau$-persistent if, with high probability, a random walk that starts at $x$ and takes $\tau$ steps in the $d$-dimensional directed hypercube ends at a vertex $y$ for which $f(y) \leq f(x)$. 
As we show, the upper bound on the number of vertices which are not $\tau$-persistent grows linearly with the distance $\tau$ \emph{and} the image size $r$. For the tester analysis, one needs to carefully choose the distance parameter $\tau$ for which many vertices are $\tau$-persistent. In particular, this value of $\tau$  also depends on the image size $r$, resulting in the linear dependence on $r$ in the query complexity of the tester.



\paragraph{Our lower bound for testing monotonicity.} We show that our monotonicity tester is optimal among nonadaptive, 1-sided error testers. 

\begin{theorem}\label{thm:lower_bound}
There exists a constant $\eps > 0$, such that for all $d, r \in \NN$, every nonadaptive, $1$-sided error $\eps$-tester for monotonicity of functions $f \colon \{0,1\}^d \to [r]$ requires $\Omega(\min(r\sqrt{d},d))$ queries.
\end{theorem}
We prove \Thm{lower_bound} (in \Sec{lower_bound}) by generalizing a construction of Fischer et al.~\cite{FLNRRS02} that showed that nonadaptive, 1-sided error monotonicity testers of Boolean functions must make $\Omega(\sqrt{d})$ queries.
Blais et al.~\cite{BBM12} demonstrated that every tester for  monotonicity  over the $d$-dimensional hypercube domain requires $\Omega(\min(d,r^2))$ queries. Our lower bound is stronger when $r \in [2, \sqrt{d}]$, although it applies only to nonadaptive, 1-sided error algorithms. 

\paragraph{Approximating the distance to monotonicity.}
Motivated by the desire to handle noisy inputs, Parnas et al.~\cite{PRR06} generalized the property testing model to tolerant testing. 
There is a direct connection between tolerant testing of a property and approximating the distance to the property with additive and multiplicative error in the sense that these problems can be reduced to each other with the right setting of parameters and have the same query complexity up to logarithmic factors (see, e.g., \cite[Claim 2]{PRR06} and \cite[Theorem A.1]{PallavoorRW20}). One clean way to state distance approximation guarantees is to replace the additive error $\alpha$ with the promise that the input function is $\alpha$-far from the property, as specified in the following definition.
A randomized {\em $c$-approximation algorithm for the distance to monotonicity,} where $c>1$, is given a parameter $\alpha\in(0,1)$ and oracle access to a function $f\colon \{0,1\}^d \rightarrow \R$ that is $\alpha$-far from monotone. It outputs an estimate $\hat{\eps}$ that, with probability at least 2/3, satisfies $\eps(f) \leq \hat{\eps} \leq c \cdot \eps(f)$. 

Fattal and Ron \cite{FR10} studied the problem of approximating the distance to monotonicity for real-valued functions over the hypergrid domain $[n]^d$. For the special case of the hypercube domain, they give an $O(d\log r)$-approximation algorithm for functions with image size $r$ that makes $\mathrm{poly}(d, 1/\alpha)$ queries. \Thm{real_tal_robust} allows us to improve on their result, by showing that the algorithm of Pallavoor et al.~\cite{PallavoorRW20} for approximating the distance to monotonicity of Boolean functions also works for real-valued functions, without any loss in the approximation guarantee.

 \begin{theorem}\label{thm:approx_dist_mono_real}
There exists a nonadaptive $O(\sqrt{d\log d})$-approximation algorithm for the distance to monotonicity that, given a parameter $\alpha \in (0, 1)$ and oracle access to a function $f\colon \{0,1\}^d \rightarrow \mathbb{R}$ that is $\alpha$-far from monotone, makes poly($d, 1/\alpha$) queries.
\end{theorem}

Pallavoor et al.\  prove that this approximation ratio is nearly optimal for nonadaptive algorithms, even for the special case of Boolean functions. We also note that, by the connection between tolerant testing and erasure-resilient testing observed by Dixit et al.~\cite{DixitRTV18},  our \Thm{approx_dist_mono_real} implies the existence of an erasure-resilient $\eps$-tester for monotonicity of functions $f: \{0,1\}^d\to\RR$ that can handle up to $\Theta(\eps/\sqrt{d\log d})$ erasures with query complexity   poly($d, 1/\eps$). The tester of Dixit et al.\ could handle only $O(\eps/d)$ erasures.
We prove \Thm{approx_dist_mono_real} in \Sec{distance-approximation}.

\subsection{Other Prior Work on Monotonicity Testing and Open Questions}
The query complexity of monotonicity testing of Boolean functions over the hypercube has been resolved for nonadaptive testers by Chen et al.~\cite{ChenDST15,ChenWX17} who proved a lower bound of $\widetilde{\Omega}(\sqrt{d})$. For adaptive testers, the best lower bound known to date is $\widetilde{\Omega}(d^{1/3})$, also shown by \cite{ChenWX17}. It is an open question whether adaptive algorithms can do better than nonadaptive ones for functions over the hypercube domain, both in the case of Boolean functions and, more generally, for functions with small image size. As we mentioned before, there is a lower bound of $\Omega(d)$ for functions with image size $\Omega(\sqrt{d})$~\cite{BBM12}.

Monotonicity testing has also been studied for functions on other types of domains, including general partially ordered domains~\cite{FLNRRS02}, with particular attention to the hypergrid domain $[n]^d$. (It has also been investigated in the context where the distance to monotonicity is the normalized $L_p$ distance instead of the Hamming distance, but we focus our attention here on the Hamming distance.)
When $d=1$, monotonicity testing on the hypergrid $[n]$ is equivalent to testing sortedness of $n$-element arrays. This problem was introduced by Ergun et al.~\cite{EKKRV00}. Its query complexity has been completely pinned down in terms of $n$ and $\eps$ by~\cite{EKKRV00,Fis04,ChSe14,Belovs18}: it is $\Theta(\frac{\log (\eps n)}{\eps})$. Pallavoor et al.~\cite{PRV18,Ramesh} considered the setting when the tester is given an additional parameter $r$, the number of distinct elements in the array, and obtained an $O((\log r)/\eps)$-query algorithm. There are also lower bounds for this setting: $\Omega(\log r)$ for nonadaptive algorithms by~\cite{BlaRY14} and $\Omega(\frac{\log r}{\log \log r})$ for all testers for the case when $r=n^{1/3}$ by~\cite{Belovs18}.

For general $d$,
Black et al.~\cite{BlackCS18,BlackCS20} gave an $\widetilde{O}(d^{5/6})$-query tester for Boolean functions $f\colon [n]^d \to \{0,1\}$.
For real-valued functions, Chakrabarti and Seshadhri~\cite{CS13,ChSe14} proved basically matching upper and lower bounds of $O((d\log  n)/\eps)$ and $\Omega((d\log n -\log \eps^{-1})/\eps)$. However, their lower bound only applies for functions with a large image. Pallavoor et al.~\cite{PRV18} gave an $O(\frac d \eps \cdot \log \frac d \eps \cdot \log r)$-query tester, where $r$, the size of the image, is given to the tester as a parameter.
It remains open whether there is an $O(\sqrt{d})$-query tester for Boolean functions on the hypergrid domain and, in particular, whether the isoperimetric inequality of \cite{KhotMS18} can be extended to hypergrids. Note that since the Boolean Decomposition (\Thm{main1}) holds for all partially ordered domains, an isoperimetric inequality for Boolean functions on hypergrids would also generalize to real-valued functions.

\section{Directed Talagrand Inequality for Real-Valued Functions}
\label{sec:directed-real-talagrand}

In this section, we use our Boolean decomposition \Thm{main1} to prove \Thm{real_tal_robust}, which easily implies the non-robust version (\Thm{directed_talagrand_real}) as we point out in the introduction. Let $f \colon \{0,1\}^d \to \RR$ be a non-monotone function over the $d$-dimensional hypercube and let $\mathtt{col} \colon \cS_f^- \to \{\red,\blue\}$ be an arbitrary 2-coloring of $\cS_f^-$. Given $x \in \{0,1\}^d$ and a subgraph $\cH$ of the $d$-dimensional hypercube, we define the quantities 
\begin{align*}
    I_{f,\red,\cH}^-(x) =& \left|\left\{ y \colon (x,y) \in \cS_f^- \cap E(\cH)\text{, } \mathtt{col}(x,y) = \red\right\}\right|; \\   
    I_{f,\blue,\cH}^-(y) =& \left|\left\{ x \colon (x,y) \in \cS_f^- \cap E(\cH) \text{, } \mathtt{col}(x,y) = \blue\right\}\right| \text{.}
\end{align*}

Let $f_1,\ldots,f_k \colon \{0,1\}^d \to \{0,1\}$ be the Boolean functions and $\cH_1,\ldots,\cH_k$ be the disjoint subgraphs of the $d$-dimensional hypercube that are guaranteed by \Thm{main1}. 
Let $C'$ denote the constant from the robust Boolean isoperimetric inequality (Theorem 2.7 of ~\cite{PallavoorRW20}) that is hidden by $\Omega$.
We have
\begin{align}
	\underset{\bx \sim \{0,1\}^d}{\EX}\left[\sqrt{I_{f,\red}^{-}(\bx)}\right] + \underset{\by \sim \{0,1\}^d}{\EX}\left[\sqrt{I_{f,\blue}^{-}(\by)}\right] &\geq \underset{\bx}{\EX}\left[\sqrt{I_{f,\red,\bigcup_{i=1}^k \cH_i}^{-}(\bx)}\right] + \underset{\by}{\EX}\left[\sqrt{I_{f,\blue,\bigcup_{i=1}^k \cH_i}^{-}(\by)}\right] \label{eq:1} \\
	&= \sum_{i=1}^k \left(\underset{\bx}{\EX} \left[ \sqrt{I_{f,\red,\cH_i}^{-}(\bx)}\right] + \underset{\by}{\EX} \left[ \sqrt{I_{f,\blue,\cH_i}^{-}(\by)}\right]\right) \label{eq:2} \\
	&\geq \sum_{i=1}^k \left(\underset{\bx}{\EX} \left[ \sqrt{I_{f_i,\red,\cH_i}^{-}(\bx)} \right] + \underset{\by}{\EX} \left[ \sqrt{I_{f_i,\blue,\cH_i}^{-}(\by)} \right]\right) \label{eq:3} \\
	&= \sum_{i=1}^k \left(\underset{\bx}{\EX} \left[ \sqrt{I_{f_i,\red}^{-}(\bx)} \right] + \underset{\by}{\EX} \left[ \sqrt{I_{f_i,\blue}^{-}(\by)} \right] \right)\label{eq:4} \\
	&\geq \sum_{i=1}^k C' \cdot \eps(f_i) \label{eq:5} \\
	&\geq \frac{C' \cdot \eps(f)}{2} \text{.} \label{eq:6}
\end{align}

The inequality \Eqn{1} holds simply because $\bigcup_{i=1}^k \cH_i$ is a subgraph of the $d$-dimensional hypercube, while the equality \Eqn{2} holds because the $\cH_i$'s are disjoint. The inequality \Eqn{3} holds since $\cS_{f_i}^- \subseteq \cS_f^-$ and the equality \Eqn{4} holds since $\cS_{f_i}^- \subseteq E(\cH_i)$ (these are both by item 2 of \Thm{main1}). Finally, \Eqn{5} is due to Theorem 2.7 of \cite{PallavoorRW20} and \Eqn{6} is due to item 1 of \Thm{main1}.

\section{Boolean Decomposition: Proof of \Thm{main1}}
\label{sec:decomposition} 

In this section, we prove the Boolean Decomposition \Thm{main1}. Our results consider any partially ordered domain, which we represent by a DAG $\cG$. The {\em transitive closure} of $\cG$, denoted $\text{TC}(\cG)$, is the graph $(V(\cG),E)$, where $E=\{(x,y) \colon x\prec y\}$. The {\em violation graph} of $f$ is the graph $(V(\cG),E')$, where $E'$ is the set of edges in $E$ violated by $f$. 

In \Sec{cover}, we define the key notion of sweeping graphs 
and identify some of their important properties. In \Sec{matching-decomp}, we prove a general lemma that shows how to use a matching $M$ in $\text{TC}(\cG)$ to find disjoint sweeping graphs in $\cG$ satisfying a ``matching rearrangement'' property. The techniques in \Sec{cover} and \Sec{matching-decomp} are inspired by the techniques of \cite{BlackCS18} used to analyze Boolean functions on the hypergrid domain, $[n]^d$. In \Sec{subgraphs}, we apply our matching decomposition lemma to a carefully chosen matching to obtain the subgraphs $\cH_1,\ldots,\cH_k$. Finally, in \Sec{functions}, we define the Boolean functions $f_1,\ldots,f_k$ and complete the proof of \Thm{main1}.





\subsection{Sweeping Graphs and Their Properties} \label{sec:cover}


\begin{definition} [$(S,T)$-Sweeping Graphs] \label{def:cover} Given a DAG $\cG$ and $s,t \in V(\cG)$, define $\cH(s,t)$ to be the subgraph of $\cG$ formed by the union of all directed paths in $\cG$ from $s$ to $t$. Given two disjoint subsets $S,T \subseteq V(\cG)$, define the {\em $(S,T)$-sweeping graph}, denoted $\cH(S,T)$, to be the union of directed paths in $\cG$ that start from some $s \in S$ and end at some $t \in T$. That is, 
\[\cH(S,T) = \bigcup_{(s,t) \in S \times T} \cH(s,t) \text{.}\]
\end{definition}
Note that if $s \npreceq t$ then $\cH(s,t) = \emptyset$. 
\smallskip

We now prove three properties of sweeping graphs which we use in \Sec{functions} to analyze our functions $f_1,\ldots,f_k$. Given disjoint sets $S,T \subseteq V(\cG)$ and $z \in V(\cH(S,T))$, define the sets
\begin{align*} \label{eq:reachable-sets}
    S(z) = \{s \in S \colon s \preceq z \} \text{ and } T(z) = \{t \in T \colon z \preceq t \} \text{.}    
\end{align*}

\begin{claim} [Properties of Sweeping Graphs] \label{clm:cg-props} Let $\cG$ be a DAG and $S,T \subseteq V(\cG)$ be disjoint sets.
\begin{enumerate}
    \item \emph{(Property of Nodes in a Sweeping Graph):} If $z \in V(\cH(S,T))$ then $S(z) \neq \emptyset$ and $T(z) \neq \emptyset$. 
    \item \emph{(Property of Nodes Outside of a Sweeping Graph):} If $z \in V(\cG) \setminus V(\cH(S,T))$ then at most one of the following is true: (a) $\exists y \in V(\cH(S,T))$ such that $z \prec y$, (b) $\exists x \in V(\cH(S,T))$ such that $x \prec z$.
    \item \emph{(Sweeping Graphs are Induced):} If $x,y \in V(\cH(S,T))$ and $(x,y) \in E(\cG)$ then $(x,y) \in E(\cH(S,T))$.
\end{enumerate}
\end{claim}

\begin{proof} 
Property 1 holds by definition of the sweeping graph $\cH(S,T)$. If $z \in V(\cH(S,T))$, then, by definition of $\cH(S,T),$ there exist $s \in S$ and $t \in T$ for which $z$ belongs to some directed path from $s$ to $t$. That is, $z \in V(\cH(s,t))$. Thus $s \in S(z)$ and $t \in T(z),$ and property 1 holds. 

We now prove property 2. Suppose, for the sake of contradiction, that there exist $x,y,z \in V(\cG)$ for which $x,y \in V(\cH(S,T))$, $z \notin V(\cH(S,T))$, and $x \prec z \prec y$. By property 1, there exist some $s \in S(x)$ and some $t \in T(y)$. Then $s \preceq x \prec z \prec y \preceq t$ and, consequently, $z$ belongs to some directed path from $s$ to $t$. Thus $z \in V(\cH(s,t)),$ and so $z \in V(\cH(S,T))$. This is a contradiction. 

We now prove property 3. Suppose $x,y \in V(\cH(S,T))$ and $(x,y) \in E(\cG)$. By property 1, there exist $s \in S$ and $t \in T$ for which $s \preceq x$ and $y \preceq t$. Since $(x,y) \in E(\cG)$, we have $x \prec y$ and so $s \preceq x \prec y \preceq t$. Thus, the edge $(x,y)$ belongs to a directed path from $s$ to $t$. That is, $(x,y) \in E(\cH(s,t))$ and so $(x,y) \in E(\cH(S,T))$.\end{proof}



\subsection{Matching Decomposition Lemma for DAGs} \label{sec:matching-decomp}

In this section, we prove the following matching decomposition lemma. Recall that $\text{TC}(\cG) = (V(\cG),E)$ denotes the transitive closure of the DAG $\cG$, where $E = \{(x,y) \colon x \prec y\}$. Consider a matching $M$ in $\text{TC}(\cG)$. We represent $M \colon S \to T$ as a bijection between two disjoint sets $S,T \subseteq V(\cG)$ of the same size for which $s\prec M(s)$ for all $s \in S$. For a set $S'\subseteq S,$ define $M(S')=\{M(s) \colon s\in S'\}$. Note that for convenience we will sometimes abuse notation and represent $M$ as the set of pairs, $\{(s,M(s)) \colon s \in S\}$, instead of as a bijection.

\begin{lemma} [Matching Decomposition Lemma for DAGs] \label{lem:general-partition} For every DAG $\cG$ and every matching $M \colon S \to T$ in $\emph{TC}(\cG)$, there exist partitions $(S_i \colon i \in [k])$ of $S$ and $(T_i \colon i \in [k])$ of $T$, where $M(S_i) = T_i$ for all $i \in [k]$, and the following hold.

\begin{enumerate} [noitemsep]
    \item \emph{(Sweeping Graph Disjointness):} $V(\cH(S_i,T_i)) \cap V(\cH(S_j,T_j)) = \emptyset$ for all $i \neq j$, where $i,j \in [k]$.
    \item \emph{(Matching Rearrangement Property):} For all $i \in [k]$ and $(x,y) \in S_i \times T_i$, if $x \prec y$ then there exists a matching $\hM \colon S_i \to T_i$ in $\emph{TC}(\cG)$ for which $(x,y) \in \hM$.
\end{enumerate}
\end{lemma}

\begin{proof} In \Alg{merge}, we show how to construct partitions $(S_i \colon i \in [k])$ for $S$ and $(T_i \colon i \in [k])$ for $T$ from a matching $M$ in $\text{TC}(\cG)$. \Alg{merge} uses the following notion of conflicting pairs.

\begin{definition} [Conflicting Pairs] \label{def:conflict} Given a DAG $\cG$ and four disjoint sets $X,Y,X',Y' \subset V(\cG)$, we say that the two pairs $(X,Y)$ and $(X',Y')$ conflict if $V(\cH(X,Y)) \cap V(\cH(X',Y')) \neq \emptyset$. \end{definition}

\begin{algorithm}
  \setstretch{1.1}
  \caption{Algorithm for constructing conflict-free pairs from a matching $M$ \label{alg:merge}}
  \begin{algorithmic}[1]
    \Require{A DAG $\cG$ and a matching $M \colon S \to T$ in $\text{TC}(\cG)$.}
    \Statex
    \Let{$\cQ_0$}{$\left\{(\{x\}, \{y\}) \colon (x,y) \in M\right\}$} \Comment{Initialize pairs using $M$}
    \For{$s \geq 0$}
        \If{two pairs $(X,Y) \neq (X',Y') \in \cQ_s$ \emph{conflict}}
          \Let{$\cQ_{s+1}$}{$\left(\cQ_s \setminus \{(X,Y), (X',Y')\}\right) \cup \{(X \cup X',Y \cup Y')\}$}
          \Comment{Merge conflicting pairs}
        \Else{}
          \Let{$s^{\ast}$}{$s$} and \Return{$\cQ_{s^{\ast}}$} \label{step:termination}
          \Comment{Terminate when there are no conflicts}
        \EndIf
    \EndFor
  \end{algorithmic}
\end{algorithm}

\begin{figure}[h]
    \begin{center}  \includegraphics[scale=0.4]{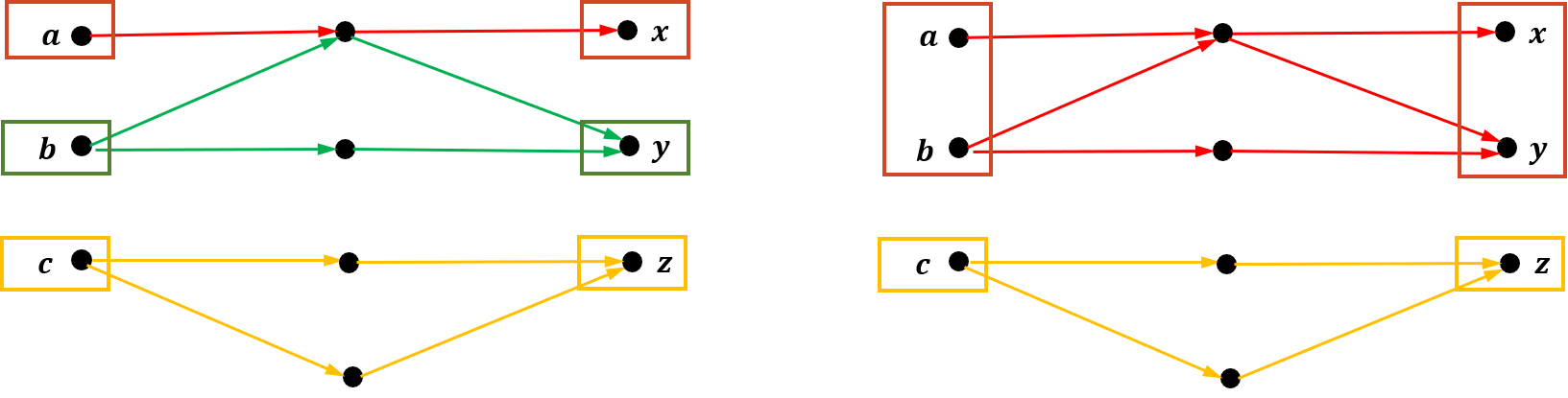}
    \end{center}
    \caption{An illustration for \Alg{merge} with input matching $M = \{(a, x), (b, y), (c, z)\}$. We initialize $\cQ_0 = \{(\{a\}, \{x\}), (\{b\}, \{y\}), (\{c\}, \{z\})\}$. The pairs $(\{a\}, \{x\})$ and $(\{b\}, \{y\})$ conflict, so we merge them to obtain a new and final collection $Q_1 = \{(\{a,b\}, \{x,y\}), (\{c\}, \{z\})\}$. }
    \label{fig:merge}
\end{figure}

The following observation is apparent and by design of \Alg{merge}.

\begin{observation} [Loop Invariants of \Alg{merge}] \label{obs:invariants} For all $s \in \{0,1,\ldots,s^{\ast}\}$, (a) $M(X) = Y$ for all $(X,Y) \in \cQ_s$, (b) $(X \colon (X,\cdot) \in \cQ_s)$ is a partition of $S$, and (c) $(Y \colon (\cdot,Y) \in \cQ_s)$ is a partition of $T$. \end{observation}

Given a matching $M \colon S \to T$ in $\text{TC}(\cG)$, we run \Alg{merge} to obtain the set $\cQ_{s^{\ast}}$. See \Fig{merge} for an illustration. Define $k = |\cQ_{s^{\ast}}|$ and let $\{(S_i,T_i) \colon i \in [k]\}$ be the set of pairs in $\cQ_{s^{\ast}}$. By \Obs{invariants}, $(S_i \colon i \in [k])$ is a partition of $S$, $(T_i \colon i \in [k])$ is a partition of $T$, and $M(S_i) = T_i$ for all $i \in [k]$. Item 1 of \Lem{general-partition} holds since \Alg{merge} terminates at step $s$ only when all pairs in $\cQ_{s}$ are non-conflicting (recall \Def{conflict}). Thus, to prove \Lem{general-partition} it only remains to prove item 2. To do so, we prove the following \Clm{rematch1}, that easily implies item 2. Note that while we only require \Clm{rematch1} to hold for the special case of $s = s^\ast$, using an inductive argument on $s$ allows us to give a proof for all $s \in \{0,1,\ldots,s^{\ast}\}$.

\begin{claim} [Rematching Claim] \label{clm:rematch1} For all $s \in \{0,1,\ldots,s^{\ast}\}$, pairs $(X,Y) \in \cQ_s$, and $(x,y) \in X \times Y$, there exists a matching $\hM \colon X \setminus \{x\} \to Y \setminus \{y\}$ in $\emph{TC}(\cG)$. \end{claim} 

\begin{proof} The proof is by induction on $s$. For the base case, if $s = 0$, then, by inspection of \Alg{merge}, for $(X, Y) \in \cQ_0$, we must have $X = \{x\}$ and $Y = \{y\}$. Thus, setting $\hM = \emptyset$ trivially proves the claim.

Now let $s > 0$. Fix some $(X, Y) \in \cQ_s$ and $(x, y) \in X \times Y$. Let $(X_1,Y_1),(X_2,Y_2) \in \cQ_{s-1}$ be the pairs of sets in $\cQ_{s-1}$ for which $x \in X_1$ and $y \in Y_2$. First, if $(X_1,Y_1) = (X_2,Y_2)$, then by induction there exists a matching $\hM' \colon X_1 \setminus \{x\} \to Y_1 \setminus \{y\}$ in $\text{TC}(\cG)$. Note that by definition of \Alg{merge}, we must have $X_1 \subseteq X$ and $Y_1 \subseteq Y$. Then the required matching is $\hM = \hM' \cup M|_{X \setminus X_1}$ where $M|_{(\cdot)}$ denotes the restriction of the original matching $M$ to the set $(\cdot)$. 
Suppose $(X_1,Y_1) \neq (X_2,Y_2)$. This is the interesting case, and we give an accompanying illustration in \Fig{rematching}. By definition of \Alg{merge}, it must be that $(X_1,Y_1)$ and $(X_2,Y_2)$ \emph{conflict} (recall \Def{conflict}) and were merged to form $X = X_1 \cup X_2$ and $Y = Y_1 \cup Y_2$. Thus, there exists some vertex $z \in V(\cH(X_1,Y_1)) \cap V(\cH(X_2,Y_2))$ and $x_1 \in X_1, y_1 \in Y_1, x_2 \in X_2, y_2 \in Y_2$ for which $x_1 \preceq z \preceq y_1 \text{ and } x_2 \preceq z \preceq y_2$.


We now invoke the inductive hypothesis to get matchings $\hM_1 \colon X_1 \setminus \{x\} \to Y_1 \setminus \{y_1\}$ and $\hM_2 \colon X_2 \setminus \{x_2\} \to Y_2 \setminus \{y\}$ in $\text{TC}(\cG)$. Observe that $x_2 \preceq z \preceq y_1$ and thus we can match $x_2$ and $y_1$. The required matching in $\text{TC}(\cG)$ is $\hM = \hM_1 \cup \hM_2 \cup \{(x_2,y_1)\}$. \end{proof}

\begin{figure}[h]
    \begin{center}
    \includegraphics[scale=0.4]{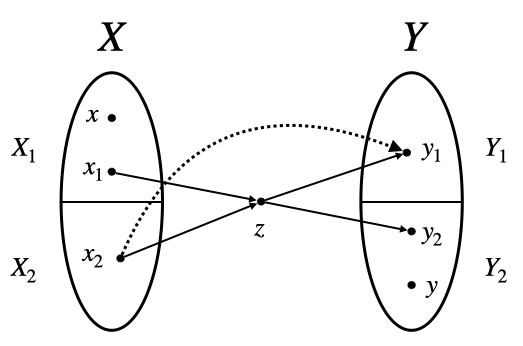}
    \end{center}
    \caption{An illustration for the case of $(X_1,Y_1) \neq (X_2,Y_2)$ in the proof of \Clm{rematch1}. The solid lines represent directed paths. 
    The dotted line represents the pair $(x_2,y_1)$ added to obtain the final matching $\hM$. The only vertices of $X \cup Y$ not participating in $\hM$ are $x$ and $y$.}
    \label{fig:rematching}
\end{figure}

We conclude the proof of \Lem{general-partition} by showing that \Clm{rematch1} implies item 2. We are given $(S_i,T_i) \in \cQ_{s^{\ast}}$ for some $i \in [k]$ and $(x,y) \in S_i \times T_i$ where $x \prec y$. By \Clm{rematch1} there exists a matching $\hM' \colon S_i \setminus \{x\} \to T_i \setminus \{y\}$ in $\text{TC}(\cG)$. We then set $\hM = \hM' \cup \{(x,y)\}$. Since $x \prec y$, the final matching $\hM \colon S_i \to T_i$ is a matching in $\text{TC}(\cG)$ which contains the pair $(x,y)$.
\end{proof}

\subsection{Specifying a Matching to Construct the Subgraphs $\cH_1,\ldots,\cH_k$} \label{sec:subgraphs}

In this section, we apply \Lem{general-partition} to a carefully chosen matching $M$ in order to construct our disjoint subgraphs $\cH_1,\ldots,\cH_k$. 


\begin{definition} [Max-weight, Min-cardinality Matching] \label{def:maxminmatching} A matching $M$ in $\emph{TC}(\cG)$ is a max-weight, min-cardinality matching if $M$ maximizes $\sum_{(x,y) \in M} (f(x) - f(y))$ and among such matchings minimizes $|M|$. 
\end{definition}


Henceforth, let $M$ denote a max-weight, min-cardinality matching. Let $S$ and $T$ denote the set of lower and upper endpoints, respectively, of $M$. We use the following well-known fact on matchings in the violation graph.

\begin{fact} [Corollary 2 \cite{FLNRRS02}] \label{fact:flnrrs} For a $\mathrm{DAG}$ $\cG$ and function $f\colon V(\cG) \to \R$, the distance to monotonicity $\eps(f)$ is equal to the size of the minimum vertex cover of the violation graph of $f$ divided by $|V(\cG)|$. 
\end{fact}

\begin{fact} \label{fact:matching} $M$ is a matching in the violation graph of $f$ that is also maximal. That is, (a) $f(x) > f(y)$ for all $(x,y) \in M$ and (b) $|M| \geq (\eps(f) \cdot |V(\cG)|)/2$. \end{fact}


\begin{proof} First, for the sake of contradiction, suppose $f(x) \leq f(y)$ for some pair $(x,y) \in M$. Then we can set $M = M \setminus \{(x,y)\}$, which can only increase $\sum_{(x,y) \in M} (f(x) - f(y))$ and will decrease $|M|$ by $1$. This contradicts the definition of $M$. Thus, $f(x) > f(y)$ for all $(x,y) \in M$ and so $M$ is a matching in the \emph{violation graph of $f$}. Second, since $M$ maximizes $\sum_{(x,y) \in M} (f(x) - f(y))$, it must also be a \emph{maximal} matching in the violation graph of $f$. Thus, (b) follows from \Fact{flnrrs} and the fact that the size of any maximal matching is at least half the size of the minimum vertex cover. \end{proof}

We now apply \Lem{general-partition} to $M$, obtaining the partitions $(S_i \colon i \in [k])$ and $(T_i \colon i \in [k])$ for $S$ and $T$, respectively, for which $M(S_i) = T_i$ for all $i \in [k]$. For each $i \in [k]$, let $\cH_i = \cH(S_i,T_i)$. We use the collection of sweeping graphs $\cH_1,\ldots,\cH_k$ to prove \Thm{main1}. Note that these subgraphs are all vertex disjoint by item 1 of \Lem{general-partition}. We use item 2 of \Lem{general-partition} to prove the following lemma regarding the $(S_i,T_i)$ pairs. The proof crucially relies on the fact that $M$ is a max-weight, min-cardinality matching. 

\begin{lemma} [Property of the Pairs $(S_i,T_i)$]
\label{lem:greater} For all $i \in [k]$ and $(x,y) \in S_i \times T_i$, if $x \prec y$ then $f(x) > f(y)$. \end{lemma}

\begin{proof} Suppose there exists $i \in [k]$, $x \in S_i$, and $y \in T_i$ for which $x \prec y$ and $f(x) \leq f(y)$. 
By item 2 of \Lem{general-partition} there exists a matching $\hM \colon S \to T$ in $\text{TC}(\cG)$ for which $(x,y) \in \hM$. 
In particular, since $M$ and $\hM$ have identical sets of lower and upper endpoints,
\[
\sum_{(s,t) \in \hM} (f(s) - f(t)) = \sum_{(s,t) \in M} (f(s) - f(t)) \text{ and } |\hM| = |M|.
\]
Now set $\hM' = \hM \setminus \{(x,y)\}$ and observe that since $f(x) \leq f(y)$, 
\[
\sum_{(s,t) \in \hM'} (f(s) - f(t)) \geq \sum_{(s,t) \in M} (f(s) - f(t)) \text{ and } |\hM'| < |M| \text{.}
\]
Therefore, $M$ is not a max-weight, min-cardinality matching and this is a contradiction. \end{proof}

%

\subsection{Tying it Together: Defining the Boolean Functions $f_1,\ldots,f_k$} \label{sec:functions}

We are now equipped to define the functions $f_1,\ldots,f_k \colon V(\cG) \to \{0,1\}$ and complete the proof of \Thm{main1}. First, given $i \in [k]$ and $z \in V(\cG) \setminus V(\cH_i)$, we say that $z$ is \emph{below} $\cH_i$ if there exists $y \in V(\cH_i)$ for which $z \prec y$, and $z$ is \emph{above} $\cH_i$ if there exists $x \in V(\cH_i)$ for which $x \prec z$. Since $\cH_i$ is the $(S_i,T_i)$-sweeping graph, by item 2 of \Clm{cg-props}, vertex $z$ cannot be both below and above $\cH_i$, simultaneously. 
Second, given $z \in V(\cH_i)$, we define the set $T_i(z) = \{t \in T_i \colon z \preceq t\}$. Note that by item 1 of \Clm{cg-props}, $T_i(z) \neq \emptyset$ for all $z \in V(\cH_i)$,  and so the quantity $\max_{t \in T_i(z)} f(t)$ is always well-defined.

\begin{definition} \label{def:f_i} For each $i \in [k]$, define the function $f_i \colon V(\cG) \to \{0,1\}$ as follows. For every $z \in V(\cG)$,
\begin{align*} \label{eq:f_i}
    f_i(z) =
    \begin{cases}
      1, & \text{if}\ z \in V(\cH_i) \emph{ and } f(z) > \max_{t \in T_i(z)} f(t), \\
      0, & \text{if}\ z \in V(\cH_i) \emph{ and } f(z) \leq \max_{t \in T_i(z)} f(t), \\
      1, & \text{if}\ z \notin V(\cH_i) \emph{ and } z \emph{ is above } \cH_i, \\
      0, & \text{if}\ z \notin V(\cH_i) \emph{ and } z \emph{ is not above } \cH_i \text{.}
    \end{cases}
\end{align*}
%
\end{definition}

See \Fig{diamond} for an illustration of the values of $f_i$. We first prove item 1 of \Thm{main1}. Recall that $M(S_i) = T_i$ for all $i \in [k]$. Let $M_i = M|_{S_i}$ denote the matching $M$ restricted to $S_i$. Consider $x \in S_i$. By \Lem{greater}, $f(x) > f(y)$ for all $y \in T_i$ such that $x \prec y$. Thus $f(x) > \max_{t \in T_i(x)} f(t)$ and so $f_i(x) = 1$. Now consider $y \in T_i$. Observe that $y \in T_i(y)$. Thus, clearly, $f(y) \leq \max_{t \in T_i(y)} f(t)$, and so $f_i(y) = 0$. Therefore, $f_i(x) = 1$ for all $x \in S_i$ and $f_i(y) = 0$ for all $y \in T_i$. In particular, $f_i(x) = 1 > 0 = f_i(M(x))$ for all $x \in S_i$ and so $M_i$ is a matching in the violation graph of $f_i$. Thus, $\eps(f_i) \geq \frac{|M_i|}{|V(\cG)|}$ for all $i \in [k]$. It follows,
\begin{align*}
    \sum_{i=1}^k \eps(f_i) \geq |V(\cG)|^{-1} \sum_{i=1}^k |M_i| = |V(\cG)|^{-1} \cdot |M| \geq |V(\cG)|^{-1} \cdot \frac{\eps(f) \cdot |V(\cG)|}{2} = \frac{\eps(f)}{2}
\end{align*}
by the above argument and \Fact{matching}. Thus, item 1 of \Thm{main1} holds. \smallskip

\begin{figure}[h]
    \begin{center}
    \includegraphics[scale=0.4]{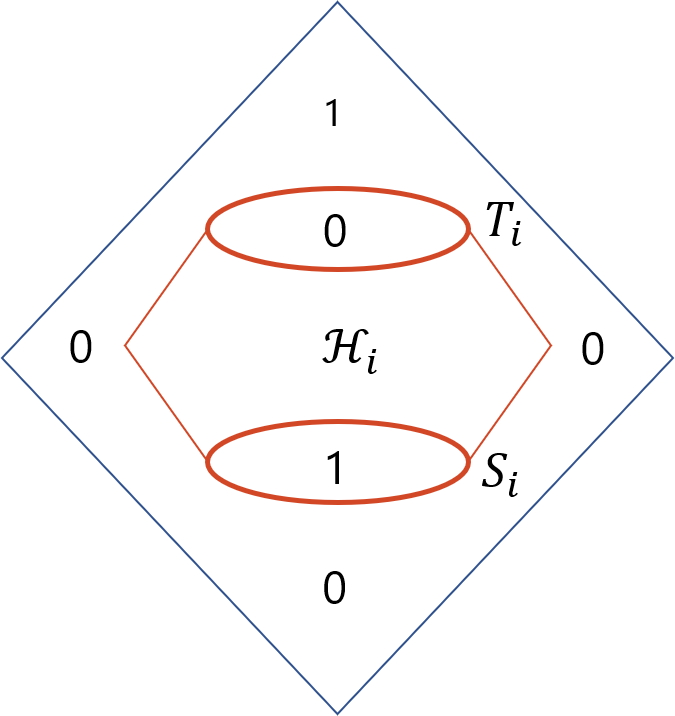}
    \end{center}
    \caption{An illustration for the Boolean function $f_i$ of \Def{f_i}. The diamond represents the DAG $\cG$ whose paths are directed from bottom to top. The hexagon represents the sweeping graph $\calH_i = \calH(S_i, T_i)$.  The value of $f_i$ is 1 for  the vertices in $S_i$ and  0 for the vertices in $T_i$. For vertices outside of $\cH_i$, its value is 1 for those vertices which are above $\calH_i$ and 0 for vertices which are not above $\cH_i$. }
    \label{fig:diamond}
\end{figure}



To prove item 2 of \Thm{main1}, we need to show that for all $i \in [k]$ the following hold: 
\[\cS_{f_i}^- \subseteq E(\cH_i) \text{ and } \cS_{f_i}^- \subseteq \cS_f^-{.}\]
We first prove that $\cS_{f_i}^- \subseteq E(\cH_i)$. Consider an edge $(x,y) \in E(\cG) \setminus E(\cH_i)$. We need to show that $f_i(x) \leq f_i(y)$. First, observe that if both $x,y \in V(\cH_i)$, then by item 3 of \Clm{cg-props}, we have $(x,y) \in E(\cH_i)$. Thus, we only need to consider the following three cases. Recall that $f_i(x),f_i(y) \in \{0,1\}$.
\begin{enumerate} [noitemsep]
    \item $x \in V(\cH_i)$, $y \notin V(\cH_i)$: In this case, $y$ is above $\cH_i$, and so $f_i(y) = 1$. Thus, $f_i(x) \leq f_i(y)$.
    \item $x \notin V(\cH_i)$, $y \in V(\cH_i)$: In this case, $x$ is below $\cH_i$, and so $x$ is \emph{not} above $\cH_i$ by item 2 of \Clm{cg-props}. Thus, $f_i(x) = 0$, and so $f_i(x) \leq f_i(y)$.
    \item $x \notin V(\cH_i)$, $y \notin V(\cH_i)$: If $x$ is above $\cH_i$, then $y$ is above $\cH_i$ as well, and so $f_i(x) = f_i(y) = 1$. Otherwise, $x$ is \emph{not} above $\cH_i$ and so $f_i(x) = 0$. Thus, $f_i(x) \leq f_i(y)$.
\end{enumerate}
Therefore, $S_{f_i}^- \subseteq E(\cH_i)$.

We now prove that $\cS_{f_i}^- \subseteq \cS_f^-$. Consider an edge $(x,y) \in \cS_{f_i}^-$. Then $f_i(x) = 1$ and $f_i(y) = 0$. Since $S_{f_i}^- \subseteq E(\cH_i)$, we have $(x,y) \in E(\cH_i)$ and so $x,y \in V(\cH_i)$. By definition of the functions $f_i$, it holds that $f(x) > \max_{t \in T_i(x)} f(t)$ and $f(y) \leq \max_{t \in T_i(y)} f(t)$. Since $x \prec y$, then $T_i(y) \subseteq T_i(x)$, because all vertices reachable from $y$ are also reachable from $x$. Therefore, 
\[f(x) > \max_{t \in T_i(x)} f(t) \geq \max_{t \in T_i(y)} f(t) \geq f(y) \text{.}\]
\noindent Thus $f(x) > f(y)$, and so $(x,y) \in \cS_f^-$. As a result, $S_{f_i}^- \subseteq \cS_f^-$ and item 2 of \Thm{main1} holds. This concludes the proof of \Thm{main1}.

\section{Testing Monotonicity of Real-Valued Functions}\label{sec:tester}

In this section, we prove \Thm{tester}. We show that the  tester of \cite{KhotMS18} for Boolean functions can be employed to test monotonicity of real-valued functions. The tester is simple: it queries two comparable vertices $x$ and $y$ and rejects if the pair exhibits a violation to monotonicity for $f$. The tester tries different values $\tau$ for the distance between $x$ and $y$, that is, the number of coordinates on which they differ. The key step in the analysis of \cite{KhotMS18} (and in our analysis) is to show that for some choice of $\tau$, the tester will detect a violation to monotonicity with high enough probability. The extra factor of $r$ in the query complexity of our tester arises because we are forced to choose $\tau$ which is a factor of $(r-1)$ smaller than for the Boolean case. Intuitively, the reason for this is that as the walk length $\tau$ increases, the probability that the function value stays below a certain threshold decreases. We make this precise in \Sec{persistence}. 

We first define the distribution from which the tester samples $x$ and $y$. Following this, we present the tester as \Alg{tester}. Let $p$ denote the largest integer for which $2^p \leq \sqrt{d / \log d}$. In \Alg{tester}, we sample pairs of vertices at distance $\tau$, where $\tau$ ranges over the powers of two up to $2^p$.

\begin{definition} [Pair Test Distribution] \label{def:pt_distribution} Given parameters $b \in \{0,1\}$ and a positive integer $\tau$, define the following distribution $\cD_{\texttt{pair}}(b,\tau)$ over pairs $(x,y) \in (\{0,1\}^d)^2$. Sample $\bx$ uniformly from $\{0,1\}^d$. Let $\bS = \{i \in [d] \colon \bx_i = b\}$. If $\tau > |\bS|$, then set $\by = \bx$. Otherwise, sample a uniformly random set $\bT \subseteq \bS$ of size $|\bT| = \tau$. Obtain $\by$ by setting $\by_i = 1-\bx_i$ if $i \in \bT$ and $\by_i = \bx_i$ otherwise. \end{definition} 



\begin{algorithm}
  \setstretch{1.3}
  \caption{Monotonicity Tester for $f \colon \{0,1\}^d \to \RR$ \label{alg:tester}}
  \begin{algorithmic}[1]
    \Require{Parameters $\eps \in (0,1)$, dimension $d$, and image size $r$; oracle access to function $f \colon \{0,1\}^d \to \RR$.}
    \Statex 
    \State \textbf{for all} $b\in\{0,1\}$ and $\tau\in\{1,2,4,\dots,2^p\}$ \textbf{do}
    \State \indent \textbf{repeat} $\widetilde{O}\left(\min\big(\frac{r\sqrt{d}}{\eps^2},\frac{d}{\eps}\big)\right)$ times:
        \State \indent\indent Sample $(\bx,\by) \sim \cD_{\texttt{pair}}(b,\tau)$.
        \State \indent\indent \textbf{if} $b=0$ and $f(\bx) > f(\by)$ \textbf{then} \textbf{reject}. \Comment{if $b=0$ then $\bx \preceq \by$}
        \State \indent\indent \textbf{if} $b=1$ and $f(\bx) < f(\by)$ \textbf{then} \textbf{reject}. \Comment{if $b=1$ then $\bx \succeq \by$}
    \State \textbf{accept}.
  \end{algorithmic}
\end{algorithm}

Our tester only uses comparisons between function values, not the values themselves. Thus, for the purposes of our analysis we can consider functions with the range $[r]$ w.l.o.g. 

When $\tau = 1$, the algorithm is simply sampling edges from the $d$-dimensional hypercube. The distribution from which we sample is not the uniform distribution on edges, but following an argument from \cite{KhotMS18}, we can assume that for $\tau=1$, our tester has the same guarantees as the edge tester.

The choice of the distance parameter $\tau$ for which the rejection probability of the tester is high depends on the existence of a certain ``good'' bipartite subgraph of violated edges. Our analysis differs from the analysis of \cite{KhotMS18} both in how we obtain the ``good'' subgraph of violated edges and in the choice of the optimal distance parameter $\tau$. 

We extend the following definitions from \cite{KhotMS18}. Let $G(A,B,E_{AB})$ 
denote a directed bipartite graph with vertex sets $A$ and $B$ and all edges in $E_{AB}$ \emph{directed} from $A$ to $B$.

\begin{definition}[$(K,\Delta)$-Good Graphs] A directed bipartite graph $G(A, B, E_{AB})$ is $(K, \Delta)$-good if for $X, Y$ such that either $X = A$, $Y=B$ or $X=B$, $Y = A$, we have: (a) $|X| = K$. (b) Vertices in $X$ have degree exactly $\Delta$. (c) Vertices in $Y$ have degree at most $2\Delta$. The graph $G$ is $(K, \Delta)$-left-good if $X=A$ and $(K, \Delta)$-right-good if $X=B$. 
\end{definition}

The \emph{weight} of $x \in \{0,1\}^d$, denoted by $|x|$, is the number of coordinates of $x$ with value 1.  

\begin{definition}[Persistence] 
Given a function $f \colon \{0,1\}^d \to [r]$ and an integer $\tau \in \left[1,\sqrt{\frac{d}{\log d}}\right]$, a vertex $x \in \{0,1\}^d$ of weight in the range $\frac{d}{2} \pm O(\sqrt{d\log d})$ is {\em $\tau$-right-persistent} for $f$ if
\[
\Pr_{\by}[f(\by) \leq f(x)] > \frac{9}{10}\,,
\]
where $\by$ is obtained by choosing a uniformly random set $\bT \subset \{ i \in [d] \colon x_i=0 \}$ of size $\tau$ and setting $\by_i = 1-x_i$ if $i \in \bT$ and $\by_i = x_i$ otherwise\footnote{Note that $\tau \geq |\{i \in [d] \colon x_i = 0\}|$ by our assumption on $x$ and $\tau$.}. We define $\tau$-left-persistence symmetrically. 
\end{definition}

We use 
the following technical claim implicitly shown in the analysis of the tester of \cite{KhotMS18}. 

\begin{claim}[\cite{KhotMS18}]\label{clm:kms_blackbox} Suppose there exists a $(K, \Delta)$-right-good subgraph $G(A,B,E_{AB})$ of the directed $d$-dimensional hypercube, such that (a) $E_{AB} \subseteq \cS_f^-$, (b) $K \sqrt{\Delta} = \Theta(\frac{\eps(f) \cdot 2^d}{\log d})$, and (c) at least $\frac{99}{100}|B|$ of the vertices in $B$ are $(\tau'-1)$-right-persistent for some $\tau'$ such that $\tau' \cdot \Delta \ll d$. Then there exists a constant $C' > 0$, such that for $(\bx, \by) \sim \cD_{\texttt{pair}}(0, \tau')$,
\[\Pr_{\bx,\by}[f(\bx) > f(\by)] \geq \frac{C' \cdot \tau'}{d} \cdot \frac{K}{2^d} \cdot \Delta\text{.}\] 
\end{claim}
The analogous claim holds given a $(K, \Delta)$-left-good subgraph with many $(\tau'-1)$-left-persistent vertices in $A$ and $(\bx,\by)$  drawn from $\cD_{\texttt{pair}}(1,\tau')$. 

In \Sec{goodsubgraph}, we prove \Lem{good_subgraph} which obtains a good subgraph for $f$ satisfying conditions (a) and (b) of \Clm{kms_blackbox}. In \Sec{persistence}, we prove \Lem{persistent} which gives an upper bound on the fraction of non-persistent vertices, enabling us to satisfy condition (c). Finally, in \Sec{final_analysis}, we use \Lem{good_subgraph} and \Lem{persistent} to show that the conditions of \Clm{kms_blackbox} are satisfied. Finally, we use it to prove \Thm{tester}.

\subsection{Existence of a Good Bipartite Subgraph} \label{sec:goodsubgraph}

In this section, we prove \Lem{good_subgraph} on the existence of good bipartite subgraphs for real-valued functions, which was proved in~\cite{KhotMS18} for the special case of Boolean functions. This lemma crucially relies on our isoperimetric inequality for real-valued functions (\Thm{real_tal_robust}). We first state (without proof) a combinatorial result of~\cite{KhotMS18}, which we need for our lemma. 

\begin{lemma}[Lemma 6.5 of \cite{KhotMS18}]\label{lem:combo} Let  $G(A, B, E_{AB})$ be a directed bipartite graph whose vertices have degree at most $2^s$. Suppose in addition, that for any 2-coloring of its edges $\mathtt{col}: E_{AB} \to \{\mathrm{red}, \mathrm{blue}\}$ we have
\begin{align}\label{eq:l_robust}
    \sum_{x \in A} \sqrt{\deg_{\mathrm{red}}(x)} + \sum_{y \in B} \sqrt{\deg_{\mathrm{blue}}(y)} \geq L, 
\end{align}
where $\deg_{\red}(x)$ denotes the number of red edges incident on $x$ and $\deg_{\blue}(y)$ denotes the number of blue edges incident on $y$. Then $G(A,B,E_{AB})$ contains a subgraph that is $(K, \Delta)$-good with $K \sqrt{\Delta} \geq \frac{L}{8s}$. 
\end{lemma}

We can now generalize Lemma 7.1 of \cite{KhotMS18}.

\begin{lemma}\label{lem:good_subgraph}
For all functions $f \colon \{0,1\}^d \to \RR$, there exists a subgraph $G(A,B,E_{AB})$ of the directed, $d$-dimensional hypercube which is $(K, \Delta)$-good, where $K \sqrt{\Delta} = \Theta(\frac{\eps(f) \cdot 2^d}{\log d})$ and $E_{AB} \subseteq \cS_f^-$. 
\end{lemma}

\begin{proof} 
Our proof relies on \Lem{combo}. Condition \Eqn{l_robust} is clearly reminiscent of the isoperimetric inequality in \Thm{real_tal_robust}.  We want to partition the vertices in $\{0,1\}^d$ into sets $A$ and $B$ such that all the violated edges are directed from $A$ to $B$ and apply  \Thm{real_tal_robust} to the resulting graph. In addition, we want \Eqn{l_robust} to hold for a big enough value of $L$.  In the Boolean case, we can simply partition the vertices by function values.  In contrast, for real-valued functions, a vertex $x \in \{0,1\}^d$ can be incident on both incoming and outgoing violated edges. 
To overcome this challenge we resort to the bipartiteness of the directed hypercube, where each edge is between a vertex with an odd weight and a vertex with an even weight.  Partition $\sfm$ into two sets:
\begin{align*}
    E_{0} &= \{(x,y) \in \cS_f^- \colon |x| \text{ is even}\}; \\
    E_{1} &= \{(x,y) \in \cS_f^- \colon |x| \text{ is odd}\}{.}    
\end{align*}
For $j \in \{0,1\}$, let $V_j$ and $W_j$ denote the set of lower and upper endpoints, respectively, of the edges in $E_j$. We consider the two subgraphs $G_j(V_j,W_j,E_j)$ for $j \in \{0,1\}$. Notice that the vertices in $V_0 \cup W_1$ have even weight and the vertices in $V_1 \cup W_0$ have odd weight. Obviously, $V_0$ and $W_1$ may not be disjoint, and similarly $V_1$ and $W_0$ may not be disjoint, and thus $G_0$ and $G_1$ may not be vertex disjoint.

We quickly explain why we cannot simply use \Lem{combo} with either $G_0$ or $G_1$. Fix a 2-coloring of the edges $E_0 \cup E_1$. By averaging, one of the graphs will have a high enough contribution to left-hand side of the isoperimetric inequality of \Thm{real_tal_robust}. Assume this graph is $G_0$. As a result, condition \Eqn{l_robust} will hold for $G_0$ with $L = \Omega(\eps \cdot 2^d)$. However, one cannot guarantee that condition \Eqn{l_robust} holds for \emph{all} possible colorings of the edges of $G_0$. Our construction below describes how to combine $G_0$ and $G_1$ so that we can jointly ``feed'' them into \Lem{combo}.



We construct copies $\widehat{G}_0$ and $\widehat{G}_1$ of $G_0$ and $G_1$, so that  $\widehat{G}_0$  contains a vertex labelled $(x, 0)$ for each vertex $x$ of $G_0$, and $\widehat{G}_1$ contains a vertex $(x, 1)$ for each vertex $x$ of $G_1$. For each edge $(x, y)$ in $G_0$ we add an edge from $(x, 0)$ to $(y,0)$ in $\widehat{G}_0$. We do the same for the edges of $G_1$. Note that each edge of $\sfm$ has exactly one copy, either in $\widehat{G}_0$ or $\widehat{G}_1$. 

Let $\widehat{G}(\widehat{V},\widehat{W},\cS_f^-)$ denote the union of the two disjoint graphs $\widehat{G}_0$ and  $\widehat{G}_1$. That is, 
\begin{align*}
    \widehat{V} &= \{ (x, 0) \: | \: x \in V_0\} \cup \{ (x, 1) \: | \: x \in V_1\}, \\
    \widehat{W} &= \{ (y, 0) \: | \: y \in W_0\} \cup \{ (y, 1) \: | \: y \in W_1\}.
\end{align*}
All the edges of $\widehat{G}$ are directed from $\widehat{V}$ to $\widehat{W}$. Although imprecise, we think of the edges of $\widehat{G}$ as $\sfm$, since each edge in $\sfm$ has exactly one copy in $\widehat{G}$. 

Consider a 2-coloring $\mathtt{col} \colon \cS_f^- \to \{\mathrm{red},\mathrm{blue}\}$. Observe that
\begin{align*}
    \sum_{(x,\cdot) \in \widehat{V}} \sqrt{I_{f,\red}^-(x)} &+ \sum_{(y, \cdot) \in \widehat{W}} \sqrt{I_{f,\blue}^-(x)} =  \sum_{x \in V_0 \cup V_1} \sqrt{I_{f,\red}^-(x)} + \sum_{y \in W_0 \cup W_1} \sqrt{I_{f,\blue}^-(y)} \\
    &=\sum_{\substack{x \in \{0,1\}^d \\ |x|  \text{ is even}}} \sqrt{I_{f,\red}^-(x)} + \sqrt{I_{f,\blue}^-(x)} +  \sum_{\substack{x \in \{0,1\}^d \\ |x|  \text{ is odd}}} \sqrt{I_{f,\red}^-(x)} + \sqrt{I_{f,\blue}^-(x)} \\
    &= \sum_{x \in \{0,1\}^d} \sqrt{I_{f,\red}^-(x)} + \sum_{y \in \{0,1\}^d} \sqrt{I_{f,\blue}^-(y)} \geq C \cdot \eps(f) \cdot 2^d,
\end{align*}
where the inequality holds by \Thm{real_tal_robust}.

%
%
By construction, $I_{f,\red}^-(x) = \deg_{\red}((x, \cdot))$ for all $(x, \cdot) \in \widehat{V}$ and $I_{f,\blue}^-(y) = \deg_{\blue}((y, \cdot))$ for all $(y, \cdot) \in \widehat{W}$. We have that condition \Eqn{l_robust} of \Lem{combo} holds with $L  = C \cdot \eps(f) \cdot 2^d$. Thus, $\widehat{G}$ contains a subgraph $G_{\text{good}}(A,B,E_{AB})$ that is $(K,\Delta)$-good with $K\sqrt{\Delta} \geq \frac{L}{8 \log d}$. Without loss of generality, assume $G_{\text{good}}(A,B,E_{AB})$ is $(K,\Delta)$-right-good. 

Let $G_{\text{good}, 0} = (A_0,B_0,E_{A_0B_0})$ denote the subgraph of $G_\text{good}$ lying in $\widehat{G}_0$ and let $G_{\text{good},1} = (A_1,B_1,E_{A_1B_1})$ denote the subgraph of $G_{\text{good}}$ lying in $\widehat{G}_1$. Since $B_0 \cap B_1 = \emptyset$, we know that either $|B_0| \geq K/2$ or $|B_1| \geq K/2$. Suppose $|B_0| \geq K/2$. Moreover, since $\widehat{G}_0$ and $\widehat{G}_1$ are vertex and edge disjoint subgraphs, the degree of a vertex of $A_0 \cup B_0$ in $G_{\text{good},0}$ is the same its degree in $G_{\text{good}}$. Thus, $G_{\text{good},0}$ is a $(K/2,\Delta)$-right-good subgraph of the $d$-dimensional directed hypercube for which $\frac{K}{2}\sqrt{\Delta} \geq \frac{L}{16\log d}$.

By removing some vertices from $B_0$, and redefining $K$ if necessary, we may assume that $K\sqrt{\Delta} = \Theta\left(\frac{\eps(f) \cdot 2^d}{\log d}\right)$. This completes the proof of \Lem{good_subgraph}. \end{proof}



\subsection{Bounding the Number of Non-Persistent Vertices} \label{sec:persistence}

We prove \Lem{persistent} that bounds the number of non-persistent vertices for a function $f$ and a given distance parameter $\tau$. All results in this section also hold for $\tau$-left-persistence. 

For a function $f\colon \{0,1\}^d \to \R$, we define $I_f^-$ as $\frac{|\sfm|}{2^d}$. 

\begin{corollary} [Corollary of Theorem 6.6, Lemma 6.8 of \cite{KhotMS18}] \label{cor:KMS_persistence} Consider a function $h \colon \{0,1\}^d \to \{0,1\}$ and an integer $\tau \in \left[1,\sqrt{\frac{d}{\log d}}\right]$. If $I_h^- \leq \sqrt{d}$ then
\begin{align} \label{eq:KMS_persist}
    \Pr_{\bx\sim \{0,1\}^d}\big[\bx \text{ is not } \tau\text{-right-persistent for } h\big] = O\left(\frac{\tau}{\sqrt{d}}\right)\text{.}
\end{align}
\end{corollary}

We generalize the above result to functions with image size $r \geq 2$.

\begin{lemma} \label{lem:persistent} 
Consider a function $f \colon \{0,1\}^d \to [r]$ and an integer $\tau \in \left[1,\sqrt{\frac{d}{\log d}}\right]$. If $I_f^- \leq \sqrt{d}$, then
\[\Pr_{\bx\sim \{0,1\}^d}\big[\bx \text{ is not } \tau\text{-right-persistent for } f \big] = (r-1) \cdot O\left(\frac{\tau}{\sqrt{d}}\right)\text{.}\]
\end{lemma}

\begin{proof} For all $t \in [r]$, define the threshold function $h_t \colon \{0,1\}^d \to \{0,1\}$ as:
\begin{align*} \label{eq:f_i}
    h_t(x) =
    \begin{cases}
      1 & \text{if }  f(x) > t,\\
      0 & \text{otherwise}.
    \end{cases}
\end{align*}
Observe that for all $t \in [r]$, we have $\cS_{h_t}^- \subseteq \sfm$, and thus $I_{h_t}^- \leq I_f^- \leq \sqrt{d}$. By \Cor{KMS_persistence}, we have that \Eqn{KMS_persist} holds for $h = h_t$ for all $t \in [r]$. Next, we point out that a vertex $x \in \{0,1\}^d$ is $\tau$-right-persistent for $f$ if and only if $x$ is $\tau$-right-persistent for the Boolean function $h_{f(x)}$. Too see this, consider a vertex $z$ such that $x \prec z$. First, note that $h_{f(x)}(x) = 0$. Second, note that $h_{f(x)}(z) = 1$ if and only if $f(z) > f(x)$ by definition of $h_{f(x)}$. Therefore, $f(z) \leq f(x)$ if and only if $h_{f(x)}(z) \leq h_{f(x)}(x)$. Finally, note that all vertices are persistent for $h_r$ since $h_r(x) = 0$ for all $x \in \{0,1\}^d$.
%
%
%
%
%
Using these observations, we have  
\begin{align*}
    \Pr_{\bx\sim\{0,1\}^d}\left[\bx \text{ is not } \tau\text{-right-persistent for } f\right] 
    &= \Pr_{\bx\sim\{0,1\}^d}\left[\bx \text{ is not } \tau\text{-right-persistent for } h_{f(\bx)}\right] \nonumber \\
    &\leq \Pr_{\bx\sim\{0,1\}^d}\left[\exists t \in [r-1] \colon \bx \text{ is not } \tau\text{-right-persistent for } h_t\right] \nonumber \\
    &\leq \sum_{t=1}^{r-1} \Pr_{\bx\sim\{0,1\}^d}\left[\bx \text{ is not } \tau\text{-right-persistent for } h_t\right] \nonumber \\
    &= \sum_{t=1}^{r-1} O\left(\frac{\tau}{\sqrt{d}}\right) = (r-1) \cdot O\left(\frac{\tau}{\sqrt{d}}\right)\,, 
\end{align*}
where the second inequality is by the union bound and the last equality is due to the fact that \Eqn{KMS_persist} holds for all $h_t$, $t \in [r]$. \end{proof}


\subsection{Proof of \Thm{tester}}\label{sec:final_analysis}

In this section, we show how to use \Lem{good_subgraph} and \Lem{persistent} to ensure that the conditions of \Clm{kms_blackbox} hold. Once the conditions are met, we prove \Thm{tester}. 

\begin{proof}[Proof of \Thm{tester}]
Let $G(A,B,E_{AB})$ be the $(K, \Delta)$-good subgraph for $f$ which we obtain from \Lem{good_subgraph}. Then $K \sqrt{\Delta} = \Theta(\frac{\eps(f)\cdot 2^d}{\log d})$ and $E_{AB} \subseteq \sfm$. Without loss of generality, suppose that $G(A,B,E_{AB})$ is a $(K,\Delta)$-right-good subgraph. Note that $G(A,B,E_{AB})$ satisfies the conditions (a) and (b) of \Clm{kms_blackbox}. We define $\sigma = K/2^d$, so that $\sigma\sqrt{\Delta} =  \Theta(\frac{\eps(f)}{\log d})$. Before proceeding with the main analysis, we rule out some simple cases with the following claim.

\begin{claim} \label{clm:parameter_bounds} Suppose any of the following hold: (a) $I_f^- \geq \sqrt{d}$. (b) $r \geq \frac{\sqrt{d}}{\log d}$. (c) $\sigma \leq \frac{r \cdot \log d}{\sqrt{d}}$. Then, for $(\bx,\by) \sim \cD_{\texttt{pair}}(0,1)$, we have $\Pr_{\bx,\by}[f(\bx) > f(\by)] \geq \widetilde{\Omega}(\frac{\eps(f)^2}{r\sqrt{d}})$. \end{claim}

\begin{proof} 
As we remarked, for $\tau = 1$, \Alg{tester} has the same guarantees as the edge tester. By definition, the edge tester rejects with probability at least $\frac{I_f^-}{d}$. Therefore, (a) implies the conclusion, since if $I_f^- \geq \sqrt{d}$, then the edge tester succeeds with probability $\Omega(\frac{1}{\sqrt{d}})$. In addition, the edge tester rejects with probability $\Omega(\frac{\eps(f)}{d})$ for all real-valued functions. Thus, (b) implies the conclusion, since if $r \geq \frac{\sqrt{d}}{\log d}$, then $\frac{\eps(f)}{d} \geq \frac{\eps(f)^2}{r\sqrt{d}\log d}$. 

To see that (c) implies the conclusion, suppose $\sigma \leq \frac{r \cdot \log d}{\sqrt{d}}$. Recall that $\sigma \sqrt{\Delta} = \Theta(\frac{\eps(f)}{\log d})$. Thus,
\begin{align*} 
    \sigma \cdot \Delta = \frac{(\sigma \sqrt{\Delta})^2}{\sigma} = \sigma^{-1} \cdot \Theta\left(\Big(\frac{\eps(f)}{\log d}\Big)^2\right) = \Omega\left(\frac{\eps(f)^2\sqrt{d}}{r (\log d)^3}\right)\text{.}
\end{align*}
Next, recall that $E_{AB} \subseteq \sfm$ and since $G$ is $(K,\Delta)$-right-good, we have $|E_{AB}| = |B| \cdot \Delta = K \cdot \Delta$. Thus, $I_f^- \geq \frac{K \Delta}{2^d} = \sigma \cdot \Delta$. Therefore, the edge tester rejects with probability $\frac{I_f^-}{d} \geq \frac{\sigma \Delta}{d} \geq \Omega\left(\frac{\eps(f)^2}{r\sqrt{d} \cdot (\log d)^3}\right)$. \end{proof}

In light of \Clm{parameter_bounds}, we henceforth assume that $I_f^- \leq \sqrt{d}$, $r \leq \frac{\sqrt{d}}{\log d}$, and $\sigma \geq \frac{r \cdot \log d}{\sqrt{d}}$. Note that this implies $\frac{r \cdot \log d}{\sqrt{d}} \leq 1$ and $\frac{r \cdot \log d}{\sqrt{d}} \leq \sigma \leq 1$. Since the tester iterates through all values of $\tau$ that are powers of $2$ and at most $\sqrt{\frac{d}{\log d}}$, we can fix the unique value $\tau'$ satisfying
\[
\tau' \leq \frac{\sigma}{r-1}\sqrt{\frac{d}{\log d}} \leq 2\tau'.
\]
Note that these bounds imply that $\tau' \geq \frac{1}{2} \cdot \sqrt{\log d}$. Moreover, since $I_f^- \leq \sqrt{d}$, we can apply \Lem{persistent} to conclude that the fraction of vertices in $\{0,1\}^d$ which are not $(\tau'-1)$-right-persistent for $f$ is at most $\frac{c \cdot \tau' \cdot (r-1)}{\sqrt{d}}$ for some constant $c > 0$. Using our upper bound on $\tau'$, this value is at most $\frac{c \cdot \sigma}{\sqrt{\log d}}\leq \frac{\sigma}{100}$ for sufficiently large $d$. Since $|B| = \sigma \cdot 2^d$, we conclude that at least $\frac{99}{100}|B|$ vertices in $B$ are $(\tau'-1)$-right-persistent. Finally, we show that $\Delta \cdot \tau' \ll d$. 
\begin{align*}
    \Delta \cdot \tau' \leq \Delta \cdot \frac{\sigma}{r-1}\sqrt{\frac{d}{\log d}} = \frac{1}{r-1} \cdot \sigma\sqrt{\Delta} \sqrt{\frac{d \Delta}{\log d}} \leq \frac{1}{r-1} \cdot \Theta\left(\frac{\eps(f)}{\log d}\right) \frac{d}{\sqrt{\log d}} \ll d, 
\end{align*}
and therefore condition (c) of \Clm{kms_blackbox} holds. We have shown that all conditions, (a), (b), and (c) of \Clm{kms_blackbox} hold. Therefore, for $(\bx,\by) \sim \cD_{\texttt{pair}}(0,\tau')$, we have
\[
\Pr_{\bx,\by}[f(\bx) > f(\by)] \geq \frac{C' \cdot \tau'}{d} \cdot \sigma \cdot \Delta \text{ for some constant } C' > 0.
\]
Using our lower bound on $\tau'$, it follows that 
\begin{align*}
    \Pr_{\bx,\by}[f(\bx) > f(\by)] \geq \frac{C' \cdot \tau' \cdot \sigma \cdot \Delta}{d} \geq \frac{1}{2} \cdot \frac{\sigma}{r-1}\sqrt{\frac{d}{\log d}} \cdot \frac{C' \cdot \sigma \cdot \Delta}{d} = \frac{C' \cdot \sigma^2 \cdot \Delta}{2(r-1)\sqrt{d\log d}}{.}
\end{align*}
Since $(\sigma \sqrt{\Delta})^2 = \Theta\left(\big(\frac{\eps(f)}{\log d}\big)^2\right)$, then:
\begin{align*}
    \Pr_{(\bx,\by)\sim \cD_{\texttt{pair}}(0,\tau')}[f(\bx) > f(\by)] \geq \frac{C'\eps(f)^2}{2(r-1)\sqrt{d}(\log d)^{5/2}} = \widetilde{O}\left(\frac{\eps(f)^2}{r\sqrt{d}}\right){.}
\end{align*}
Therefore, $\widetilde{O}(\frac{r\sqrt{d}}{\eps(f)^2})$ iterations of the tester with $(\bx,\by) \sim \cD_{\texttt{pair}}(0,\tau')$ will suffice for the tester to detect a violation to monotonicity and reject with high probability. This concludes the proof of \Thm{tester}. 
\end{proof}

\section{Approximating the Distance to Monotonicity of Real-Valued Functions}\label{sec:distance-approximation}
In this section, we prove \Thm{approx_dist_mono_real} by showing that the algorithm of Pallavoor et al.~\cite{PallavoorRW20} can be employed to approximate distance to monotonicity of real-valued functions. 

To prove \Thm{approx_dist_mono_real}, it is sufficient to give a tolerant tester for monotonicity of functions $f:\{0,1\}^d\to\RR$. A {\em tolerant} tester for monotonicity gets two parameters $\eps_1,\eps_2\in(0,1)$, where $\eps_1<\eps_2$, and oracle access to a function $f$. It has to accept with probability at least 2/3 if $f$ is $\eps_1$-close to monotone and reject with probability at least 2/3 if $f$ is $\eps_2$-far from monotone.  Our tester distinguishes functions that are $\widetilde{O}(\eps/\sqrt{d})$-close to monotone from those that are $\eps$-far. Suppose this tolerant tester has query complexity $q(\eps,d)$. Then, by~\cite[Theorem A.1]{PallavoorRW20}, it can be converted to a distance approximation algorithm with the required approximation guarantee and query complexity $O(q(\alpha,d)\log\log(1/\alpha)).$ The following lemma, proved by Pallavoor et al.\ for the special case of Boolean functions, states our result on tolerant testing of monotonicity. Together with the conversion procedure from tolerant testing to distance approximation discussed above, it implies \Thm{approx_dist_mono_real}.

\begin{lemma}
\label{lem:approx_dist_mono_tolerant}
There exists a fixed universal constant $c \in (0,1)$ and a nonadaptive algorithm, $\mathtt{ApproxMono}$, that gets a parameter $\eps \in (0, 1/2)$ and oracle access to a function $f\colon \{0,1\}^d \rightarrow \mathbb{R}$, makes $\mathrm{poly}(n, 1/\eps)$ queries and returns $\mathtt{close}$ or $\mathtt{far}$ as follows:
\begin{enumerate}
    \item If $\eps(f) \leq \frac{c \cdot  \eps}{\sqrt{d \log d}}$ it outputs $\mathtt{close}$ with probability at least 2/3.
     \item If $\eps(f) \geq \eps$ it outputs $\mathtt{far}$ with probability at least 2/3.
\end{enumerate}
\end{lemma}
\begin{proof}
We show that Algorithm $\mathtt{ApproxMono}$ of Pallavoor et al.~\cite{PallavoorRW20}, presented as \Alg{approx_mono}, works for real-valued functions.
At a high level, the algorithm uses the fact that a function that is far from monotone violates many edges or has a large matching of violated edges of a special type.  The first subroutine estimates the number of edges violated by the function by sampling edges uniformly at random and checking if they violate monotonicity. The second subroutine estimates the size of the special type of matching of violated edges. If either of these estimates is large enough, the algorithm outputs $\mathtt{far}$. Otherwise, it outputs $\mathtt{close}$. 
 
 The class of matchings sought by the algorithm is parametrized by a subset of the coordinates $S \subseteq [d]$. The special property of these matchings is that 
 one can verify locally whether a given point is matched by querying its neighbors and their neighbors. 

 To estimate the size of the matching parametrized by $S$, the algorithm estimates the probability of the following event $\mathtt{Capture}(x, S, f)$. We denote by $x^{(i)}$ the point in $\{0,1\}^d$ whose $i$-th coordinate is equal to $1-x_i$ and the remaining coordinates are the same as in $x$. 

\begin{definition}[Capture Event]\label{def:capture}
For a function $f\colon \{0,1\}^d \rightarrow \mathbb{R}$, a set $S \subseteq [d]$, and a point $x \in \{0,1\}^d$, let $\mathtt{Capture}(x, S, f)$ be the following event:
\begin{enumerate}
\item There exists an index $i \in S$ such that, for  $y = x^{(i)}$, the edge between $x$ and $y$ is violated by $f$. (Note that the edge between $x$ and $y$ is $(x,y)$ if $x_i=0$; otherwise, it is $(y,x).$)
\item Neither the edges of the form $(y, y^{(j)})$ nor the edges of the form  $(y^{(j)},y)$, where  $j \in S \setminus \{i\}$, are violated by $f.$ 
\end{enumerate}
Denote  $\Pr_{\bx \sim \{0,1\}^d}\left[\mathtt{Capture}(\bx, S, f) \right]$ by $\mu_f(S)$.
\end{definition}

Observe that $\mu_f(S)$ can be estimated nonadaptively, by sampling vertices $x$ uniformly and independently at random and querying $f$ on $x$ and all points that differ from $x$ in at most two coordinates.

\begin{figure}[h]
\centering
    \includegraphics[scale=.6]{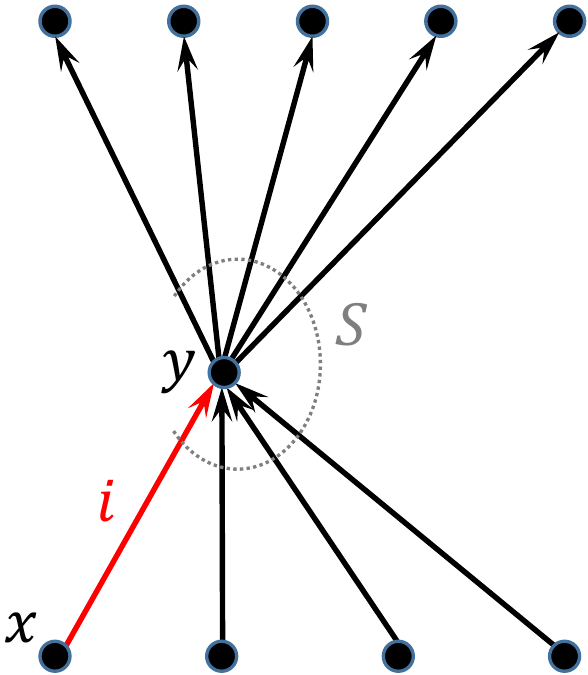}
        \hspace*{2.5cm}
    \includegraphics[scale=.6]{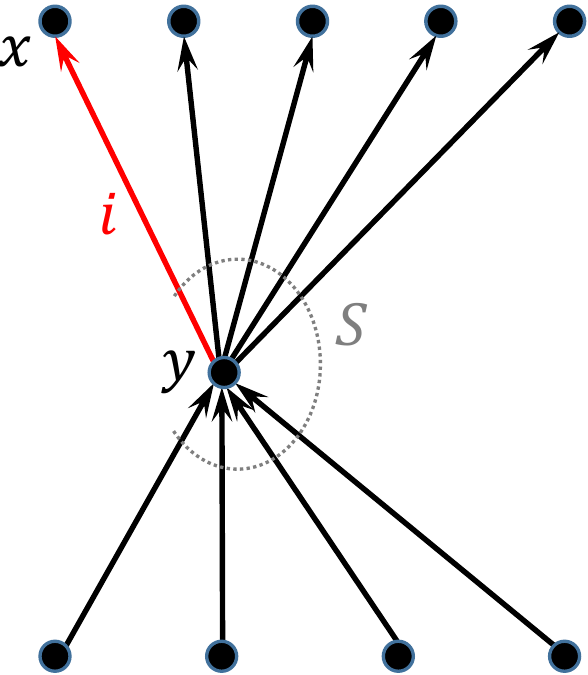}
    \caption{An illustration to \Def{capture}.  Two cases are depicted: when $x\prec y$ and when $y\prec x$.}
    \label{fig:capture-event}
\end{figure}

\begin{algorithm}
  \setstretch{1.1}
  \caption{Algorithm $\mathtt{ApproxMono}$ \label{alg:approx_mono}}
  \begin{algorithmic}[1]
    \Require{Parameters $\eps \in (0, 1/2)$ and dimension $d$; oracle access to function $f\colon \{0,1\}^d \to \mathbb{R}$}.
    \Statex
    \State Calculate $\hat{\nu}$, an estimate of the fraction of the hypercube edges that are violated by $f$, up to an additive error $ \frac{\eps}{4\sqrt{d \log d}}$.  
    \If{$\hat{\nu} \geq 3\eps/(4\sqrt{d\log d})$}{\: return $\mathtt{far}$.} \EndIf \For{$t \in \{1, 2, 4, \dots, 2^{\lfloor \log_2 d\rfloor}\}$}\label{step:params-for-s}
        \State Sample $S \subseteq [d]$ by including each coordinate $i \in [d]$ independently with probability $1/t$. 
        \State Calculate  $\hat{\mu}$, an estimate of
           $\mu_f(S)=\Pr_{\bx \sim \{0,1\}^d}\left[\mathtt{Capture}(\bx, S, f) \right]$
        up to an additive error $\frac{c' \cdot \eps}{4\sqrt{d \log d}}$ for some constant $c' > 0$. 
        \If{$\hat{\mu} \geq \frac{3c' \cdot \eps}{4\sqrt{d\log d}}$}{\:return $\mathtt{far}$}. \EndIf
    \EndFor
    \State Return $\mathtt{close}$. 
  \end{algorithmic}
\end{algorithm}


The first component of the analysis is the observation that both the fraction of violated edges and $\mu_f(S)$, for every $S\subseteq[d],$ provide a good lower bound on the distance to monotonicity. We state this observation without proof because the proof for the Boolean case from~\cite{PallavoorRW20} extends to the general case verbatim. Intuitively, it tells us that, assuming that the two estimates computed by \Alg{approx_mono} are accurate, if one of the estimates is large enough then the input function is far from monotone.
\begin{observation}[\cite{PallavoorRW20}]\label{obs:lb-on-dist}
For every function $f:\{0,1\}^d\to\RR$, the distance $\eps(f)$ is at least half the fraction of the hypercube edges that are violated by $f$ and $\eps(f)\geq \mu_f(S)/2$ for all $S\subseteq[d].$
\end{observation}
The second (and the main) component of the analysis for the Boolean case is \cite[Lemma 2.8]{PallavoorRW20}, which relies on the robust isoperimetric inequality of~\cite{KhotMS18}. 
We generalize this lemma to real-valued functions in \Lem{lemma_2.6} below. Intuitively, it states that, if function $f$ violates few edges then, for one of the $O(\log d)$ choices of the parameter tried in Step~\ref{step:params-for-s} of  \Alg{approx_mono} for sampling set $\bS$, the expectation of $\mu_f(\bS)$ is large in terms of $\eps(f)$. That is, again assuming that
the estimates computed by \Alg{approx_mono} are accurate, if none of the estimates is large enough then the input function is close to monotone.

Equipped with \Obs{lb-on-dist} and \Lem{lemma_2.6}, it is easy to convert the intuition above into the formal proof that the algorithm satisfies the guarantees of \Lem{approx_dist_mono_tolerant}. This part of the proof uses standard techniques and is the same as for the case of Boolean functions described in~\cite{PallavoorRW20}, so we omit it. This completes the proof of \Lem{approx_dist_mono_tolerant}.
\end{proof}

It remains to prove the following lemma, which crucially relies on our robust isoperimetric inequality for real-valued functions.  
We generalize the quantities used by Pallavoor et al.\ so that the proof is syntactically similar to that for the case of Boolean functions. One subtlety that arises in the case of real-valued functions is that a vertex can be incident to violated edges of both colors. In constrast, in the case of Boolean functions, each vertex can be adjacent either to violated edges going to higher-weight vertices or to violated edges going to lower-weight vertices, that is, it cannot be incident on both blue and red violated edges.

\begin{lemma}[Generalized Lemma 2.8 of \cite{PallavoorRW20}] \label{lem:lemma_2.6}
Let $f\colon \{0,1\}^d \rightarrow \mathbb{R}$ be $\eps$-far from monotone, with fraction of violated edges smaller than  $\frac{\eps}{\sqrt{d \log d}}$. Then, for some $t \in \{1, 2, 4, \dots, 2^{\lfloor \log_2 d\rfloor}\}$, it holds
$$
\underset{{\substack{\bS \subseteq [d] \\ i \in \bS \text{ w.p. } 1/t}}}{    \EX}
[ \mu_f(\bS)]
= \Omega\left( \frac{ \eps}{\sqrt{d \log d}}\right).
$$
\end{lemma}
\begin{proof}
For $x \in \{0,1\}^d$, let $U_f^-(x)$ denote the number of violated edges incident on $x$ (both incoming and outgoing).  Consider the following 2-coloring of the edges in $\sfm$: 
\[ \mathtt{col}((x, y)) = \left\{ \begin{array}{cc} 
\mathrm{red} & \text{ if  } U_f^-(x) \geq U_f^-(y); \\
\mathrm{blue} & \text{ if  } U_f^-(x) < U_f^-(y).
\end{array} \right. \]
This coloring ensures that, in the isoperimetric inequality, each edge is counted towards the endpoint incident on the largest number of violated edges (and, in case of a tie, towards the lower endpoint).

The proof of \cite[Lemma 2.8]{PallavoorRW20} relies on the existence of a set $B \subseteq \{0,1\}^d$ and a color $b \in \mathrm{\{red, blue\}}$ that satisfy the following two properties:
\begin{enumerate}
    \item no edge violated by $f$ has both endpoints in the set $B$; 
    \item  $\frac{1}{2^d} \sum_{x \in  B} \sqrt{I_{f, b}^-(x)} = \Omega(\eps)$.
\end{enumerate}
To obtain the set $B$ and the color $b$, we partition $\{0,1\}^d$ into two sets: 
\begin{align*}
    B_{\mathrm{even}} &= \{x\in \{0,1\}^d :  |x| \textnormal{ is even} \}\,, \\
    B_{\mathrm{odd}} &= \{x\in \{0,1\}^d : |x| \textnormal{ is odd} \}\,.
\end{align*}
The sets $B_{\mathrm{even}}$ and $B_{\mathrm{odd}}$ clearly satisfy property 1.
Note that, for the case of Boolean functions, Pallavoor et al.\ partition the domain points according to their function values instead of the parity of their weight to
guarantee property 1.

By \Thm{real_tal_robust},
\begin{equation*}
 \sum_{x \in  B_{\mathrm{even}}} \sqrt{I_{f, \mathrm{red}}^-(x)} + \sqrt{I_{f, \mathrm{blue}}^-(x)} +     \sum_{x \in  B_{\mathrm{odd}}} \sqrt{I_{f, \mathrm{red}}^-(x)} + \sqrt{I_{f, \mathrm{blue}}^-(x)} \geq C \cdot \eps \cdot 2^d\:. 
\end{equation*}
By averaging, there exist a color $b\in\{\mathrm{red, blue}\}$ and a set $B\in\{B_{\mathrm{even}}, B_{ \mathrm{odd}}\}$ that satisfy
\begin{equation}\label{eq:talagrand-objective-for-B}
      \sum_{x \in  B} \sqrt{I_{f, b}^-(x)} \geq \frac{C}{4} \cdot \eps \cdot 2^d\:. 
\end{equation}
Therefore, property 2 also holds. Note that due to the partition into even-weight and odd-weight points we loose an extra factor of 2 as compared to Pallavoor et al.~in the contribution of the set $B$ and the color $b$ to the isoperimetric inequality. This results in a loss by a factor of 2  (hidden in the $\Omega$-notation) in the lower bound in \Lem{lemma_2.6}. 

The rest of the proof is the same as in \cite{PallavoorRW20}, so we only summarize the key steps. We proceed by partitioning the points $x \in B$ into buckets $B_{t, s}$ for $t, s \in \{1, 2, 4, \dots, 2^{\lfloor \log_2 d\rfloor}\}$, where $t \geq s$, as follows:
\[
B_{t,s} = \{ x \in B \: : \: t \leq U_f^-(x) < 2t \textnormal{ and }   s \leq I_{f, b}^-(x) < 2s \}\,.
\]
Each vertex $x \in B_{t, s}$ is incident on between $t$ and $2t$ violated edges and between $s$ and $2s$ edges colored $b$, which are counted towards $x$ in property 2. 

 When the set $\bS$ is chosen so that each coordinate is included with probability $1/t$, it holds for all $x \in B_{t,s}$ that the event $\mathtt{Capture}(x, \bS, f)$ occurs with probability $\Omega(s/t)$. Using this claim, one can lower bound the contribution of each bucket towards $\EX_{\bS \subseteq [d]}[ \mu_f(\bS)].$
By combining the contributions of the buckets with the same value $s$ and applying the Cauchy-Schwartz inequality, one obtains
\begin{align}\label{eq:final-dist-approx}
    \sum_{t\in\{1,2,4,\dots,2^{\floor{\log_2 d}}\}} \ \ \ 
    \underset{{\substack{\bS \subseteq [d] \\ i \in \bS \text{ w.p. } 1/t}}}{    \EX}
[ \mu_f(\bS)]
= \Omega\left(\frac{1}{2^d} \cdot  \frac{(\sum_{t,s: t\geq s}|B_{t,s}|\sqrt{s})^2}{\sum_{t,s: t\geq s}|B_{t,s}|t}\right).
\end{align}
We lower bound the sum in the numerator using \Eqn{talagrand-objective-for-B} and upper bound the sum in the denominator using the assumed upper bound on the number of violated edges. As a result, we get that the left-hand side of \Eqn{final-dist-approx} is 
$\Omega(\eps\sqrt{\log d}/\sqrt{d})$. Averaging over the $O(\log d)$ possible values of $t$ yields \Lem{lemma_2.6}.
 \end{proof}

\section{Our Lower Bound for Testing Monotonicity} \label{sec:lower_bound}
In this section, we prove \Thm{lower_bound} which gives a lower bound on the query complexity of testing monotonicity of real-valued functions with 1-sided error nonadaptive testers. Fischer et al. proved \Thm{lower_bound} for the special case of $r=2$ ~\cite[Theorem 19]{FLNRRS02}. Our proof of \Thm{lower_bound} is a natural extension of their construction to the more general case of $r \in [2,\sqrt{d}]$. 

\begin{proof} 
Fix $r \in [2,\sqrt{d}]$. We show that every nonadaptive, 1-sided error tester for functions over $\{0,1\}^d$ with image size $r$ must make $\Omega(r\sqrt{d})$ queries. This implies \Thm{lower_bound}, since Blais et al.~\cite[Theorem 1.6]{BBM12} proved an $\Omega(\min(d,r^2))$ lower bound for all testers. 

For convenience, assume $d$ is an odd perfect square and $r$ divides $2\sqrt{d}+1$. We partition the points $z\in\{0,1\}^{d-1}$ into levels, according to their weight~$|z|$. We group levels from the middle of the $(d-1)$-dimensional hypercube into $r$ blocks of width $w$, where $w= \frac{2\sqrt{d}+1}{r}.$
Specifically, for each $j \in [r]$, we define the set
\[
Z_j = \left\{z \in \{0,1\}^{d-1} \colon (j-1)w \leq |z| - \Big(\frac{d-1}{2} - \sqrt{d}\Big) \leq j w \right\} \text{.}
\]
Observe that
\[\bigcup_{j=1}^r Z_j = \left\{z \in \{0,1\}^{d-1} \colon -\sqrt{d} \leq \Big||z| - \frac{d-1}{2}\Big| \leq \sqrt{d}\right\}\]
and $Z_j$ is a block of $w$ consecutive levels from the middle of the $(d-1)$-dimensional hypercube. For each $i \in [d]$, we define function $f_i \colon \{0,1\}^d \to [r]$ as follows. 
For $x \in \{0,1\}^d$ and $i \in [d]$, let $x_{-i}$ be the point in $\{0,1\}^{d-1}$ obtained by removing the $i$'th coordinate from $x$. Given $x \in \{0,1\}^d$, we define
\[
f_i(x) =
\begin{cases}
r & \text{if $|x_{-i}| > \frac{d-1}{2} + \sqrt{d}$},\\
1 & \text{if $|x_{-i}| < \frac{d-1}{2} - \sqrt{d}$}, \\
j + (1 - x_i) & \text{if $x_{-i} \in Z_j$} \text{.}
\end{cases}
\]
\begin{claim} \label{clm:distance} 
For all $i \in [d]$, $\eps(f_i) = \Omega(1)$. 
\end{claim}
\begin{proof} 
Consider the matching of edges $M = \left\{(x,y) \colon x_i = 0 \text{, } y_i = 1 \text{, and } x_{-i} = y_{-i} \in \bigcup_{j=1}^r Z_j \right\}$. Observe that all pairs in $M$ are edges violated by $f_i$ and $|M| = \Omega(1) \cdot 2^d$. 
\end{proof}

Every 1-sided error tester must accept if the function values on the points it queried are consistent with a monotone function. 
We say that a set $Q\subseteq \{0,1\}^d$ of queries {\em contains a violation for a function $f$} if there exist $x,y\in Q$ such that $x\prec y$ and $f(x)>f(y)$. 
If $Q$ does not contain a violation, then the function values on $Q$ are consistent with a monotone function.

\begin{claim} \label{clm:prob} 
For all sets $Q \subseteq \{0,1\}^d$ of queries,
\[\big|\{i \in [d] \colon Q \text{ contains a violation for } f_i\}\big| < w\cdot|Q|\text{.}\] 
\end{claim}
\begin{proof} 
We use the following claim due to \cite{BCPRS17}.

\begin{claim}[Lemma 3.18 of \cite{BCPRS17}, rephrased] \label{clm:capture} 
Let $c,d \in \NN$ and $Q \subseteq \{0,1\}^d$. Given $x,y \in Q$, define $\texttt{cap}_c(x,y)$ as follows. If $x$ and $y$ differ on at least $c$ coordinates, then let $\texttt{cap}_c(x,y)$ be the set of the first $c$ coordinates on which $x$ and $y$ differ. Otherwise, let $\texttt{cap}_c(x,y)$ be the set of all coordinates on which $x$ and $y$ differ. Define $\texttt{cap}_c(Q) = \bigcup_{x,y \in Q} \texttt{cap}_c(x,y)$. Then $|\texttt{cap}_c(Q)| \leq c(|Q|-1)$. 
\end{claim}

By design of $f_i$, if $Q$ contains a violation for $f_i$, then there exist $x,y \in Q$ that differ in at most $w$ coordinates, one of which is $i$. 
Then $i \in \texttt{cap}_w(x,y)$ and thus $i \in \texttt{cap}_w(Q)$. Therefore, by \Clm{capture},
\[\big|\{i \in [d] \colon Q \text{ contains a violation for } f_i\}\big| \leq |\texttt{cap}_w(Q)| \leq w(|Q|-1) < w \cdot |Q| \text{.}\]
This completes the proof of \Clm{prob}.
\end{proof}

Now, consider a nonadaptive tester $T$ with $1$-sided error that makes $q=q(\eps, d, r)$ queries. Let $\boldsymbol{Q} \subseteq \{0,1\}^n$ denote the random set of queries of size $q$ made by $T$. Using linearity of expectation and \Clm{prob}, 
\[
\sum_{i=1}^d \Pr[T \text{ finds a violation for } f_i] = \underset{\boldsymbol{Q}}{\EX} \Big[\big|\{i \in [d] \colon \boldsymbol{Q} \text{ contains a violation for } f_i\}\big|\Big] < w \cdot q
\] 
and therefore there exists $i \in [d]$ such that
%
\[\Pr\left[T \text{ finds a violation for } f_i\right] < \frac{w \cdot q}{d} = \frac{(2\sqrt{d}+1) \cdot q}{r d} < \frac{3q}{r\sqrt{d}} \text{,}\]
whereas, if $T$ is a valid monotonicity tester, then we must have $\Pr[T \text{ finds a violation for } f_i] \geq 2/3$. Therefore, for $T$ to be a valid monotonicity tester, we require that it makes $q \geq \frac{2}{9}r\sqrt{d} = \Omega(r\sqrt{d})$ queries.
\end{proof}

\section{Undirected Talagrand Inequality for Real-Valued Functions}\label{sec:undirected_real}

In this section we prove \Thm{undirected_real} through a simple reduction to Talagrand's (undirected) isoperimetric inequality for Boolean functions, which we state as \Thm{undirected_boolean}. 

\begin{proof} 
Given $t \in \RR$, let $p_t = \frac{1}{2^d}|\{x \colon f(x) = t\}|$ denote the fraction of points $x$ in $\{0,1\}^d$ with $f(x) = t$. Note that $\sum_{t\in \RR} p_t = 1$ and that $p_t>0$ for at most $2^d$ values of $t$. Choose $m \in \RR$ to be the smallest real number such that $\sum_{t \leq m} p_t \geq 1/2$. Then we also have $\sum_{t < m} p_t < 1/2$ and so $\sum_{t \geq m} p_t > 1/2$. 

Since $\sum_{t < m} p_t + \sum_{t > m} p_t = 1 - p_m$, only one of the two sums can be less than $\frac{1-p_m}{2}$. We define a Boolean function $h \colon \{0,1\}^d \to \{0,1\}$ as follows, depending on the value of $\sum_{t < m} p_t$. 

\begin{enumerate}
    \item Suppose $\sum_{t < m} p_t < \frac{1-p_m}{2}$. Then $\sum_{t > m} p_t > \frac{1-p_m}{2}$. Moreover, $\sum_{t \leq m} p_t \geq 1/2 \geq \frac{1-p_m}{2}$ by our choice of $m$. In this case, we define the Boolean threshold function $h: \{0,1\}^d\rightarrow \{0,1\}$ as $h(x) = 1$ if $f(x) > m$ and $h(x) = 0$ otherwise.

    \item Suppose $\sum_{t < m} p_t \geq \frac{1-p_m}{2}$. We know that $\sum_{t \geq m} p_t \geq 1/2 \geq \frac{1-p_m}{2}$ by our choice of $m$. 
    In this case, we define the Boolean threshold function $h: \{0,1\}^d\rightarrow \{0,1\}$ as $h(x) = 1$ if $f(x) \geq m$ and $h(x) = 0$ otherwise. 
    
\end{enumerate}
In either case, observe that the fraction of $0$'s and $1$'s for $h$ are both at least $\frac{1-p_m}{2}$. Thus, 
\[
\dist(h,\mathbf{const}) \geq \frac{1-p_m}{2} \geq \frac{1-\max_{t\in \RR} p_t}{2} = \frac{\dist(f,\mathbf{const})}{2}\text{.}
\] 
%
Moreover, for all edges $\{x,y\}$, if $h(x) > h(y)$ then $f(x) > f(y)$. Thus, $I_f(x) \geq I_h(x)$ for all $x$, and so
\[
\underset{\bx \sim \{0,1\}^d}{\EX}\left[\sqrt{I_f(\bx)}\right] 
\geq \underset{\bx \sim \{0,1\}^d}{\EX}\left[\sqrt{I_h(\bx)}\right] \text{.}
\]
Using these two facts and applying \Thm{undirected_boolean} yields
\begin{align*}
    \underset{\bx \sim \{0,1\}^d}{\EX}\left[\sqrt{I_f(\bx)}\right] 
    \geq \underset{\bx \sim \{0,1\}^d}{\EX}\left[\sqrt{I_h(\bx)}\right] 
    \geq \frac{\dist(h,\mathbf{const})}{\sqrt{2}} 
    \geq \frac{\dist(f,\mathbf{const})}{2\sqrt{2}}\,,
\end{align*}
and this completes the proof. 
\end{proof}

\medskip

\paragraph{Acknowledgments.} We thank Ramesh Krishnan Pallavoor Suresh for useful discussions.

\bibliographystyle{alpha}
\bibliography{biblio}

\appendix

\end{document}